\documentclass[journal,11pt,onecolumn,draftclsnofoot]{IEEEtran}
\IEEEoverridecommandlockouts
\usepackage{gensymb}
\usepackage{cite}
\usepackage{graphicx}
\usepackage{url}
\usepackage{amsmath}
\usepackage{latexsym,amssymb,epsfig,color,verbatim}
\usepackage{algorithm}
\usepackage[noend]{algpseudocode}
\usepackage{pifont}
\usepackage{caption}
\usepackage{slashbox}
\usepackage[tight,footnotesize,center]{subfigure}
\usepackage{color}
\usepackage{cancel}
\usepackage{multirow}
\usepackage{amsthm}
\usepackage{multirow}

\newcommand{\BEQA}{\begin{eqnarray}}
\newcommand{\EEQA}{\end{eqnarray}}

\newcommand{\red}[1]{\textcolor{red}{ #1}}

\newtheorem{theorem}{{\bf Theorem}}
\newtheorem{corollary}{{\bf Corollary}}
\newtheorem{definition}{{\bf Definition}}
\newtheorem{rem}{{\bf Remark}}
\newtheorem{example}{{\bf Example}}

\renewcommand\IEEEkeywordsname{Keywords}

\begin{document}
\bstctlcite{IEEEexample:BSTcontrol}
%\title{\red{Enhancing Energy Efficiency  with QoI-Guarantee}}%\red{Energy Trade-off between Communication, Computation and Storage with QoI-Guarantee in Wireless Sensor Networks}}
%\title{E$3$C$3$: Enhancing Energy Efficiency among Communication, Computation and Caching with QoI-Guarantee }
\title{{Optimal Energy  {Consumption with } Communication, Computation, Caching {and} QoI-Guarantee}}

\author{Faheem Zafari$^{1,*}$, Jian Li$^{2,*}$, Kin K. Leung$^3$, Don Towsley$^4$ and Ananthram Swami$^5$\\\small{ ${}^{1,3}$Imperial College London, ${}^{2,4}$University of Massachusetts Amherst, ${}^5$U.S. Army Research Laboratory\\${}^{1,3}$\{faheem16, kin.leung\}@imperial.ac.uk, ${}^{2,4}$\{jianli, towsley\}@cs.umass.edu, ${}^5$ananthram.swami.civ@mail.mil\\${}^*$Co-primary authors}}
\maketitle

\begin{abstract}	
Energy efficiency is a fundamental requirement of modern data communication systems, and its importance is reflected in much recent work on performance analysis of system energy consumption.  However, most works have only focused on communication and computation costs, but do not account for caching costs.  Given the increasing interest in cache networks, this is a serious deficiency.  {In this paper, we consider the problem of energy consumption in data communication, compression and caching (C$3$) with a Quality of Information (QoI) guarantee in a communication network. {Our goal is to identify the optimal data compression rate and data placement over the network to minimize the overall energy consumption in the network.}  %Our goal is to identify the optimal compression rate  and cache placement  for minimizing the overall energy consumption in the network, which helps in characterizing and analyzing the energy performance of the communication network.
	The formulated problem is a \emph{Mixed Integer Non-Linear Programming} (MINLP) problem with non-convex functions, which is NP-hard in general. }
{We} propose a variant of spatial branch and bound algorithm (V-SBB), that can {provide} the $\epsilon$-global  optimal solution to {our problem}.  
{We numerically show that our C3 optimization framework can improve the energy efficiency up to 88\% compared to  any C2 optimization between communication and computation or caching. Furthermore, for our energy consumption problem, V-SBB {provides comparatively better solution than some other MINLP solvers.}}
\end{abstract}

%\vspace{06pt}
\begin{IEEEkeywords}
	Energy Tradeoff, Data analytics, Data Caching, Quality of Information, Mixed Integer Non-Linear Programming, Spatial Branch-and-Bound
\end{IEEEkeywords}

\section{Introduction}\label{sec:intro}

The rapid growth of smart environments, and advent of Internet of Things (IoT) have led to the generation of large amounts of data. However, it is a daunting task to transmit enormous data through traditional networks due to limited bandwidth and energy limitations  \cite{nazemi2016qoi}.  
These data need to be efficiently compressed, transmitted, and cached to satisfy the Quality of Information (QoI) required by end users.  In fact, many wireless components operate on limited battery power supply and are usually deployed in remote or inaccessible areas, which necessitates the need for  designs that can enhance the  energy efficiency of the system with a QoI guarantee. 

A particular example of modern systems that require high energy efficiency is the wireless sensor network (WSN). Consider a WSN with various types of sensors, which can generate enormous amount of data to serve  end users.  On one hand, data compression has been adopted to reduce transmission (communication) cost at the expense of computation cost.  On the other hand, caches can be used as a mean of reducing transmission costs and access latency, thus  enhancing QoI but with the expense of the added caching cost.  Hence, there exists a tradeoff in energy consumption due to  data communication, computation and caching.  {This raises the question: what is the right balance between compression and caching so as to minimize the total energy consumption of the network?}

In this paper, we formulate an optimization problem {to find the optimal data compression rate and data placement to minimize the energy consumed due to data compression, communication and caching with QoI guarantee in a communication network. The formulated problem is a Mixed Integer Non-Linear Programming problem with non-convex functions, which is NP-hard in general. We propose a variant of spatial branch and bound algorithm that guarantees $\epsilon$-global\footnote{{$\epsilon$-global optimality means that the obtained solution is within $\epsilon$ tolerance of the global optimal solution.}} optimality. }

 Each node has the ability to compress and cache the data with some finite storage capacity.  We focus on wireless sensor networks as our motivating example. In particular, as shown in Figure~\ref{fig:example}, we assume that only edge sensors generate data, and there exists a single sink node that collects and serves the requests for {the data generated in this network}. {The model can be extended to include any arbitrary node that produces data at the expense of added notational complexity.}

\begin{figure}
\centering
	\includegraphics[width=0.8\linewidth]{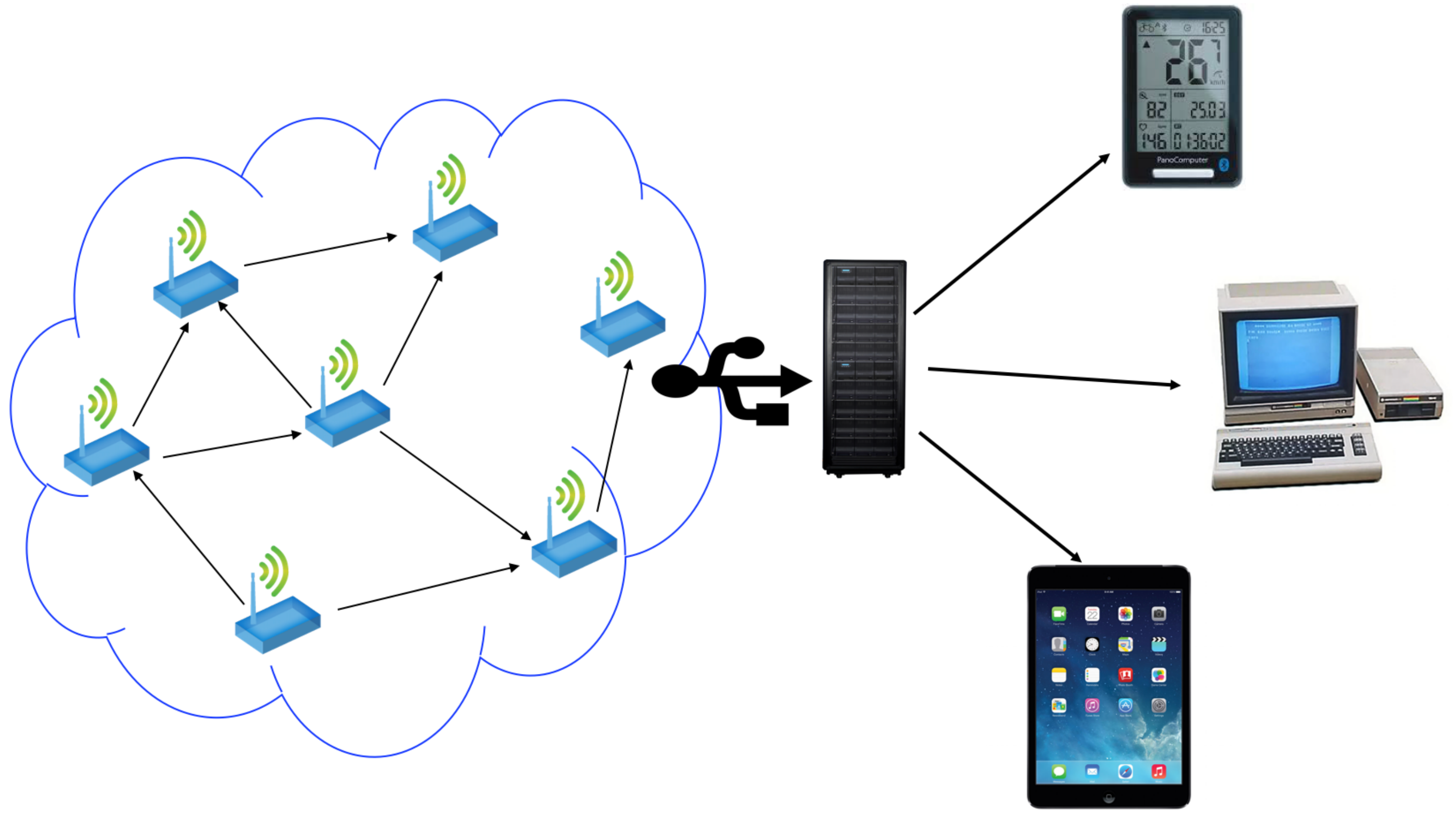}
	\vspace{-0.1in}
	\caption{A general wireless sensor network.}%The distributed model for a typical data analytics service consists of a single central and multiple edge servers in a hub-and-spoke architecture.}
\label{fig:example}
%\vspace{-0.25in}
\end{figure} 

\noindent{\textit{Computation:}} 
Data aggregation \cite{rajagopalan06, fasolo07} is the process of gathering data from multiple generators (e.g., sensors), compressing them to eliminate redundant information and then  providing the summarized  information to end users.  Since only part of the original data is transmitted, data aggregation can conserve a large amount of energy.  A common assumption in previous works is that energy required to compress data is smaller than that needed to transmit data.  {Therefore, data compression was considered a viable technique for reducing energy consumption. }
 However, it has been shown \cite{barr2006energy} that computational energy cost can be significant and may cause a net-energy increase if  data are compressed beyond a certain threshold. Hence, it is necessary to consider both  transmission and computation costs, and it is important to characterize the trade-off between them\cite{nazemi2016qoi}.

\noindent{\textit{Caching:}} Caches have been widely used in networks and distributed systems to improve performance by storing information locally, which jointly reduces access latency and bandwidth requirements, and hence improves user experience. Content Distribution Networks (CDNs), Software Defined Networks (SDNs),  Named Data Networks (NDNs) and Content Centric Networks (CCNs) are important examples of such systems. The fundamental idea behind caching is to make information available at a location closer to the end-user.  Again, most  previous work focused on designing caching algorithms to enhance system performance without considering the energy cost of caching. Caching can reduce the transmission energy by storing a local copy of the data at the requesting node (or close by), hence eliminating the need for multiple retransmission from the source node to the requesting node.   However, caching itself can incur significant energy costs \cite{choi2012network}.  Therefore, {analyzing the impact of caching on overall energy consumption in the  network (along with data communication and compression) is critical for system design.} 

\noindent{\textit{Quality of Information (QoI):}} The notion of QoI required by end users is affected by many factors. In particular, the degree of the data aggregation in a system is crucial for QoI. 
It has been shown that data aggregation can deteriorate QoI in some situations \cite{ehikioya1999characterization}. 
Thus an energy efficient design for appropriate data aggregation with a guaranteed QoI is desirable. 

We focus on a tree-structured sensor network where each leaf node generates  data, and compresses and transmits the data to the sink node in the network, which serves the requests for these data from devices outside this network.  Examples of such a setting are military sites, wireless sensors or societal networks, where a large number of devices gather data, and desire to transmit the local information to any device outside this network that requires this information.  {The objective of our work is to obtain optimal data compression rate at each node, and an optimal data placement in the network for minimizing energy consumption with QoI guarantee.}

% \vspace{-0.2in}
% \subsection{Related Work}
% \input{02-related}
% 
 
 \vspace{-0.2in}
 \subsection{Organization and Main Results}
{Section \ref{sec:rel} presents a review of relevant literature.} In Section~\ref{sec:model}, we describe our system model in which nodes are logically arranged as a tree. Each node receives and compresses data from its children node(s). The compressed data are transmitted and further compressed towards the sink node.  Each node can also cache the compressed data locally.  In Section~\ref{sec:opt}, we formulate the problem of energy-efficient data compression, communication and caching with QoI constraint as a {MINLP problem with non-convex functions}, which is {NP-hard} in general. We then show that there exists an equivalent problem obtained through symbolic reformation \cite{smith1996global} in Section~\ref{sec:sBNB}, and propose a variant of the Spatial Branch-and-Bound (V-SBB) algorithm to solve it.  We show that our proposed algorithm can achieve  $\epsilon$-global optimality.% of the original MINLP. 

In Section~\ref{sec:results}, we evaluate the performance of our optimization framework {and show that the use of caching along with data compression and communication can significantly improve the energy efficiency of a communication network. {More importantly,  we observe that with the joint optimization of data communication, computation and caching (C$3$),  energy efficiency can be improved by as much as $88\%$ compared to only optimizing communication and computation, or communication and caching (C$2$). The improvement  depends on the values of parameters in the model and the magnitude of improvement varies with different energy costs of the model.   While the improvement in energy efficient is important, our framework helps in characterizing and analyzing the enhancement in energy efficiency for different network settings.  } We also evaluate the performance of the }  proposed V-SBB algorithm  through extensive numerical studies.  In particular, we make a thorough comparison with other MINLP solvers Bonmin \cite{bonami2008algorithmic}, NOMAD \cite{le2011algorithm}, Matlab's genetic algorithm (GA), {Baron \cite{tawarmalani2005polyhedral}, SCIP \cite{achterberg2009scip} and Antigone \cite{misener2014antigone}} under different network scenarios.  The results show that our algorithm can achieve  $\epsilon$-global optimality,  and {the achieved objective function value (we achieve a lower objective function value for a minimization problem) is mostly better than stochastic algorithms such as NOMAD, GA while it performs comparably with deterministic algorithms such as Baron, Bonmin, SCIP and Antigone.
	}%is either comparable to or outperforms Bonmin.  
Furthermore, our algorithm {provides a solution in varying network situations even when other solvers such as Bonmin, and SCIP  are not able to. } % Furthermore, 
   We provide concluding remarks in Section~\ref{sec:con}.

\subsection{Related Work}\label{sec:rel}

%A key focus of this work is to demonstrate and validate the joint application of data computation and caching in systems such as WSNs, to {minimize} energy consumption due to data communication, computation and caching with a QoI guarantee.  Key decisions for such systems are how much of the computation should be performed at each server and where the data should be cached in the system.  The basic building blocks of our model {are} simple and have been studied in various settings.  However, 
To the best of our knowledge,   there is no prior work that jointly considers communication, computation and caching costs in {distributed} networks with a QoI guarantee for end users. %One of the important contributions of this paper is to develop an optimization algorithm that minimizes the total system energy costs by characterizing the tradeoff between communication, computation and caching costs with a QoI guarantee for end users. 

\noindent{\textit{\textbf{Data Compression:}}} Compression is a key operation in modern communication networks and has been supported by many data-parallel programming models \cite{boykin14}.  For WSNs, data compression is usually performed over a hierarchical topology to improve communication energy efficiency \cite{rajagopalan06}, whereas we focus on energy tradeoff between communication, computation and caching. 

\noindent{\textit{\textbf{Data Caching:}}}  {Caching} plays a significant role in many systems with hierarchical topologies, e.g., WSNs, microprocessors, CDNs etc.  There is a rich literature on the performance of caching in terms of designing different caching algorithms, e.g., \cite{jian17,ioannidis16}, and we do not attempt to provide an overview here.  %Utility maximization approach has also been studied for cache management \cite{dehghan16,nitishjian17,nitishjianfaheem17,nitishjiantech18}.  
However, none of these work considered the costs of caching, which may be significant in some systems \cite{choi2012network}. The recent paper by Li et al. \cite{jianfaheem18icpe} is closest to the problem we tackle here.  The differences between our work and \cite{jianfaheem18icpe} are mainly from two perspectives. First, the mathematical formulations are quite different, we consider energy tradeoffs between C3 while \cite{jianfaheem18icpe} focused on C2.  Second, we provide a $\epsilon$-optimal solution to a MINLP problem while \cite{jianfaheem18icpe} aimed at developing approximation algorithms.

% \red{Dimokas et al. \cite{dimokas2008cooperative} presented a cooperative caching protocol designed for WSNs. Dimokas et al. \cite{dimokas2011high} also proposed two cooperative caching protocols for WSNs that minimized latency and improved the energy efficiency of the WSN.  Alipio et al. \cite{alipio2017cache} presented a detailed survey on cache based transport protocols for WSNs. } \blue{[\bf{not sure if this is the right position. May motivate it in the introduction. I will do that. }]}

\noindent{\textit{\textbf{Energy Costs:}}}  While optimizing energy costs in wireless sensor networks has been extensively studied \cite{heinzelman2000energy}, existing work primarily is concerned with routing \cite{manjeshwar2001teen}, MAC protocols \cite{heinzelman2000energy}, and clustering \cite{ye2005eecs}. With the growing deployment of smart sensors in modern systems \cite{nazemi2016qoi}, in-network data processing, such as data aggregation, has been widely used as a mean of reducing system energy cost by lowering the data volume for transmission.

\section{Analytical Model}\label{sec:model}
 
  \begin{figure}
	\centering
	\includegraphics[width=0.8\textwidth]{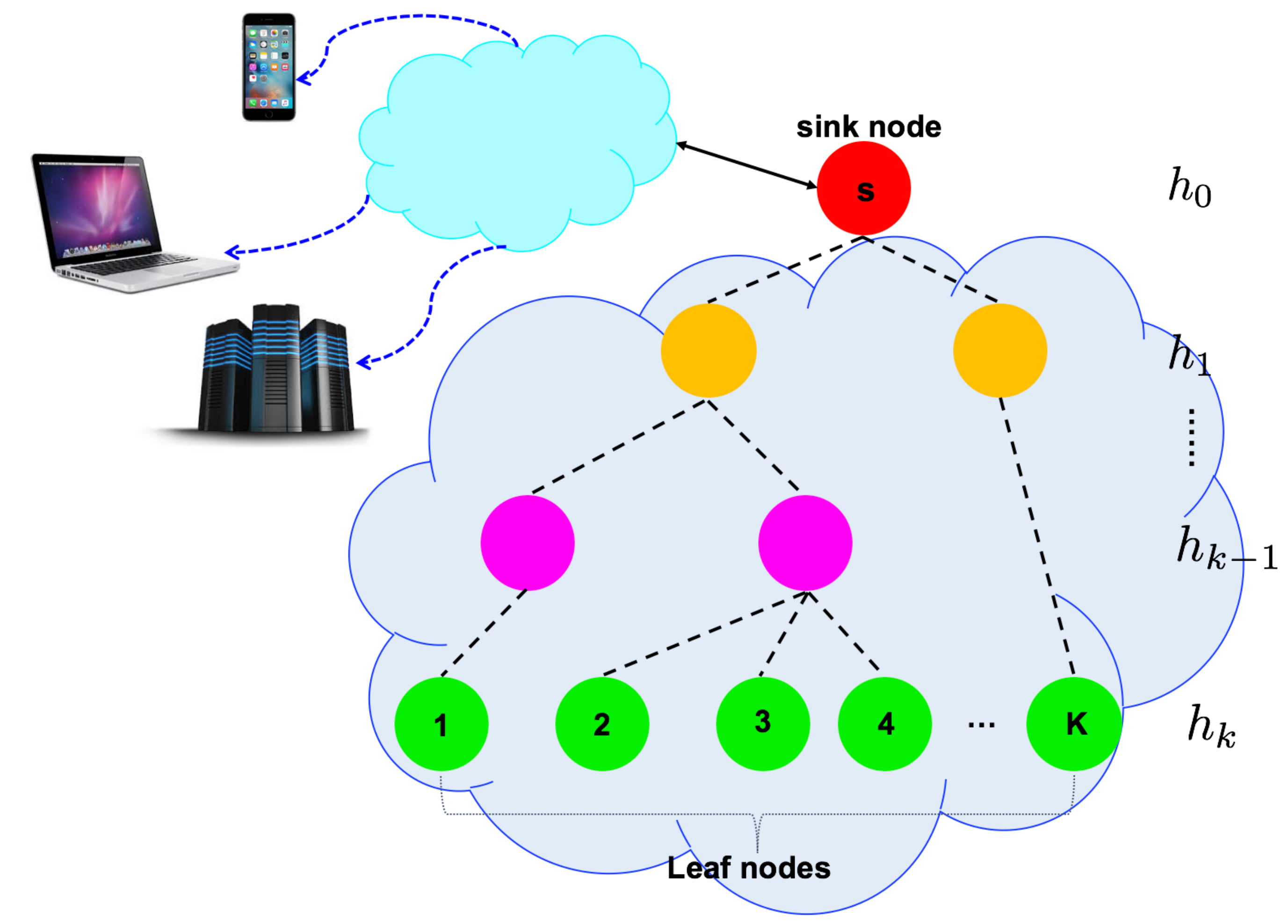}
	\vspace{-0.15in}
	\caption{Tree-Structured Network Model.}
	\protect\label{fig:1}
%	\vspace{-0.25in}
\end{figure}

We represent the network as a directed graph $G=(V, E).$ For simplicity, we consider a tree, with $N=|V|$ nodes, as shown in Figure~\ref{fig:1}. {It is possible to generalize our framework to general network topology with arbitrary source nodes, provided that the  route between the source and requesting node is known. } Node $v \in V$ is capable of storing $S_v$ amount of data.  Let $\mathcal{K}\subseteq V$ with $K=|\mathcal{K}|$ be the set of leaf nodes, i.e., $\mathcal{K}=\{1, 2, \cdots, K\}$.  {Time is partitioned in periods of equal length $T>0$ and data generated in each period are independent.  Without loss of generality (W.l.o.g.), we consider one particular period in the remainder of the paper.}  We assume that only leaf nodes $k\in\mathcal{K}$ can generate data, and all other nodes in the tree receive and compress data from their children nodes, and either cache or transmit the compressed data to their parent nodes {during time T}. {Arbitrary source nodes can also be incorporated into the model at the cost of added notational and model complexity. }%\red{During the time period $T$,  each node just produces a specific amount of  data that is cached  and served based on the requests. }  %We assume that only the leaf nodes $k\in\mathcal{K}$ generate the data which is represented as $y_{k}$. 
% There is a sink/root node at depth 0 i.e. for a sink s, $h(s)=0$. The leaf nodes are at depth/height h(k). 	We define the path $\mathcal{P}^k$ of length $h(k)$ as a sequence $\{p_1, p_2, \cdots, p_{|h(k)|}\}$ of nodes $p_j\in V$ such that $(p_j, p_{j+1})\in E.$ The data $y_k$ generated by leaf node $k$ can only be cached in one node on the path $\mathcal{P}^k$. Each node can receive (sense in the case of leaf nodes), transmit, compress and cache the data. 

Let $y_k$ be the amount of data generated by leaf node $k\in\mathcal{K}$. The data generated at the leaf nodes are transmitted up the tree to sink node $s,$ which serves requests for data generated in the network.  Let $h(k)$ be the depth of node $k$ in the tree. W.l.o.g., we assume that the sink node is located at level $h(s)=0.$  We represent a path from node $k$ 
 %leaf node
  to the sink node as the unique path $\mathcal{H}^k$ of length $h(k)$ as a sequence {$\{h_0^k, h_1^k, \cdots, h_{h(k)}^k\}$} of nodes $h_j^k\in V$ such that $(h_j^k, h_{j+1}^k)\in E,$ where $h_0^k\triangleq s$ (i.e., the sink node) and $h_{h(k)}^k\triangleq k$ (i.e., the node itself).  %\red{We assume that the data $y_k$ generated by leaf node $k$ can  be cached only at one node on the path $\mathcal{H}^k$ for a fixed duration $T$. \bf [comment: I think we should drop this as Ananthram has suggested that multiple copies do not make any sense. What do you think? ]}
% {\bf [comment: Regarding the path that Don has pointed out, I think we can just tell him that we sticked to the notation used in the PIMRC paper. That is why it is currently like that]}
 
%We assume that all the nodes $v \in V$ have the same per-bit transmission cost $\varepsilon_{vT}$, reception cost $\varepsilon_{vR}$, and compression cost $\varepsilon_{vC}$, which can be easily generalized to different costs for different nodes.  We denote $S_v$ as the storage capacity at node $v\in V.$ 

We denote the per-bit reception, transmission and compression cost of node $v \in V$ as $\varepsilon_{vR}, \varepsilon_{vT}$, and $\varepsilon_{vC},$ respectively.  Each node $h_i^k$ along the path $\mathcal{H}^k$ can compress the data generated by leaf node $k$ with a \emph{data reduction rate} $\delta_{k,i}$, where $0\leq\delta_{k, i}\leq 1,$ $\forall i, k.$ %where $\delta=\frac{Volume\; of\; Ouptput\; Data}{Volume\; of\; Input\; Data}$.} 
The reduction rate characterizes the degree to which a node can compress the received data, which plays an important role for determining the QoI.

%As mentioned earlier, energy is a critical factor that we need to take into consideration in networks.  The network may prefer to reduce energy consumption by receiving and transmitting higher compressed data.  However, 
The higher the value of $\delta_{k, i}$, the lower the compression will be, and vice versa. The higher the degree of data compression, the larger will be the amount of energy consumed by compression.  Similarly, caching the data  closer to the sink node may reduce the transmission cost for serving the request, however, each node only has finite storage capacity. We study the trade-off among the energy consumed at each node for transmitting, compression and caching the data.

Denote the total energy consumption at node $v$ as $E_v$, which consists of reception cost $E_{vR}$,  transmission cost $E_{vT}$, computation cost $E_{vC}$ and storage (caching) cost $E_{vS}$; it takes the form 
\vspace{-0.05in}
\begin{align}\label{eq:energy}
&E_v=E_{vR}+E_{vT}+E_{vC}+E_{vS},\nonumber\displaybreak[0]\\
%\end{align}
%\vspace{-0.05in}
%where 
%\vspace{-0.05in}
%\begin{align}
\text{where} \quad&E_{vR}=y_{v}\varepsilon_{vR}, \quad E_{vT}=y_{v}\varepsilon_{vT}\delta_{v}, \nonumber\displaybreak[0]\\
&E_{vC}=y_{v}\varepsilon_{vC}l_v(\delta_v),  \quad E_{vS}=w_{ca}y_{v}T. 
\end{align}
 {The above energy consumption models for data transmission, compression and caching have been used in literature \cite{nazemi2016qoi,laporta2012,choi2012network} and are suitable for highlighting the energy consumption in a communication network. However, our formulation can be extended to incorporate various other energy consumption models as well.  In \eqref{eq:energy}}, $l_v(\delta_v)$ captures the computation energy. As computation energy increases with the degree of compression, we assume that $l_v(\delta_v)$ is a continuous, decreasing and differentiable function of the reduction rate.  One candidate function is $l_v(\delta_v)=1/\delta_{v}-1$ \cite{nazemi2016qoi,laporta2012}.  Moreover, we consider an energy-proportional model \cite{choi2012network} for caching, i.e., $E_{vS}=w_{ca}y_{v}T$ if the received data $y_v$ is cached for a duration of $T$  where $w_{ca}$ represents the power efficiency of caching,  which strongly depends on the storage hardware technology. W.l.o.g., $w_{ca}$ is assumed to be identical for all the nodes.
For simplicity, denote {$f(\delta_v)$}= $\varepsilon_{vR}$+$\varepsilon_{vT}\delta_v$+$\varepsilon_{vC}l_v(\delta_v)$ as the sum of per-bit reception, transmission and compression cost at node $v$ per unit time.

During  time period  $T$, we assume that there are $R_{k}$ requests at sink node $s$ for data $y_k$ generated by leaf node $k$. For simplicity, we assume that the number of requests for the data of a node $k$ is constant.  The boolean variable $b_{k,i}$ equals $1$ if the data from node $k$ is stored along the path $\mathcal{H}^k$  at node $h_i^k,$ otherwise it equals $0$. We allow the data to be cached at only one node along the unique path between the leaf node and root node.  For ease of notation, we define $b_{k, h(k)}$ by $b_k.$  Let $C_v$ denote the set of leaf nodes $k \in \mathcal{K}$ that are descendants of node $v$.
%Furthermore, we assume that the network is structured i.e. the sink node knows which unique path to follow for a specific request. 
We also assume that the {energy cost for searching for data at different nodes} in the network is negligible \cite{nazemi2016qoi, ioannidis16}.
For convenience, let $f_{k,h(k)} \triangleq f_k$ and $\delta_{k,h(k)}\triangleq \delta_k.$  For  ease of exposition,  the parameters used throughout this paper are summarized in Table~\ref{tab:notations}.
\begin{table}
	%	\tiny
	\centering
	%\caption{List of parameters used in the paper}
	\caption{Summary of notations}
	\vspace{-0.1in}
	%\begin{tabular}{|l|p{6cm}||l|p{6cm}|}
	\begin{tabular}{|l|p{11cm}|}
		\hline
		\textbf{Notation} & \textbf{Description} \\ \hline
		$y_{k}$ 	&  number of data (bits) generated at node $k$ \\ \hline
		$\delta_{k,v}$	& reduction rate at node $v$, is the ratio of amount of output data to input data  \\ \hline
		$\gamma$	& the QoI threshold    \\ \hline
		 $\varepsilon_{vR}$ & per-bit reception cost of node $v$   \\ \hline
		$\varepsilon_{vT}$ & per-bit transmission cost of node $v$   \\ \hline
		   $\varepsilon_{vC}$& per-bit compression cost of node $v$ \\ \hline
		%	$a_{k,i}(t)$	& will be 1 if the node i transmits the data of k at time t; otherwise 0  
	%	$f_k(\delta_k)$ & $\varepsilon_{iR}$ + $\delta_k\varepsilon_{iT}$+$\varepsilon_{iC}(1/\delta_k-1)$ \\ \hline 
		 	$b_{k, v}$	& $1$ if node $v$ caches the data from leaf node $k$; otherwise $0$ \\ \hline
		$S_v$ & storage capacity of node $v$ \\ \hline
		 $w_{ca}$ & caching power efficiency  \\  \hline
		$R_k$ & request rate for data from node $k$\\ \hline
		%  $a_{k,v}$ & $1$ if node $v$ transmits data from leaf node $k$; otherwise $0$\\  \hline
		$N$ & total number of nodes in the network \\ \hline
		 $C_v$ &set of leaf nodes that are descendants of node $v$  \\  \hline
		$T$	& time length that data are cached\\ \hline
		 $\phi^u$ & upper bound of the objective function  \\ \hline
		$\mathcal{L}$	& list of  regions \\ \hline
		 $\mathcal{R}$ & any sub-region in $\mathcal{L} $  \\ \hline
		$\phi^{\mathcal{R}, u}$	& upper bound on the objective function in subregion $\mathcal{R}$ \\ \hline
		 $\phi^{\mathcal{R}, l}$	& lower bound on the objective function in subregion $\mathcal{R}$  \\ \hline
	     $\epsilon$	& difference between the upper and lower bound   \\ \hline
	     $w_i^{\mathcal{R},l}$ & lower bound on auxiliary variable $w_i$ in subregion $\mathcal{R}$ \\ \hline
	    $w_i^{\mathcal{R},u}$ & upper bound on auxiliary variable $w_i$ in subregion $\mathcal{R}$ \\ \hline
	    $w_{bc}^{j}$ & $j^{\text{th}}$ candidate variable for branching  \\ \hline
	      $w_{b}$ & chosen branching variable  \\ \hline
	       $w_{b}^\mathcal{r}$& value at which the variable is branched  \\ \hline
	        $\text{bt}$ &  bilinear terms \\ \hline
	        $\text{lft}$ & linear fractional terms \\ \hline
	        $\mathcal{T}_{\text{bt}}$ & set of bilinear terms (\text{bt})\\ \hline
	         	         $\mathcal{T}_{\text{lft}}$ & set of linear fractional terms (\text{lft})\\ \hline 
		%	$y_{R,k}$ & The number of bits per request for data of node k & &  \\ \hline 
		%	$f_k(\sigma_k)$ & $\varepsilon_{iR}$ + $\sigma_k\varepsilon_{iT}$+$\varepsilon_{iC}(\frac{1}{\delta_k}-1)$   & &   \\  \hline
	\end{tabular}
	\label{tab:notations}
	\vspace{-0.15in}
\end{table}

\section{Energy Optimization}\label{sec:opt}

In this section, we first define the cost function in our model and then formulate the optimization problem. Data produced by every leaf node is received, transmitted, and possibly compressed by all nodes in the path from  the leaf  node to the root node, consuming energy  
\begin{align}
E^{\text{C}}_k= \sum_{i=0}^{h(k)}y_kf(\delta_{k, i})\prod_{m=i+1}^{h(k)}\delta_{k,m},
\label{eq:servingcost}
\end{align}
where $\prod_{m=i}^{j} \delta_{k,m} := 1$ if $i \ge j$.  Equation~\eqref{eq:servingcost} captures one-time\footnote{During every time period  $T$,  data is always pushed towards the sink upon the first request.} energy cost of receiving, compressing and transmitting data $y_k$ from leaf node (level $h(k)$) to the sink node (level $0$).  The amount of data received by any node at level $i$ from leaf node $k$ is $y_k\prod_{m=i+1}^{h(k)}\delta_{k,m}$ due to the compression from level $h(k)$ to $i+1.$  The term $f(\delta_{k,i})$ captures the reception, transmission and compression energy cost for node at level $i$ along the path from leaf node $k$ to the sink node.  

Let $E_k^{\text{R}}$ be the total energy consumed in responding to the subsequent $(R_k-1)$ requests.  We have
\begin{align}
E^{\text{R}}_k&= \sum_{i=0}^{h(k)}y_{k}(R_k-1)\Bigg\{f(\delta_{k, i})\prod_{m=i+1}^{h(k)}\delta_{k,m}\bigg(1-\sum_{j=0}^{i}b_{k,j}\bigg)+ \bigg(\prod_{m=i}^{h(k)}\delta_{k,m}\bigg)b_{k,i}\left(\frac{w_{ca}T}{R_k-1}+\varepsilon_{kT}\right)\Bigg \}.
\label{eq:cachretrievecost}
\end{align}
{Note that the remaining $(R_k-1)$ requests are either served by the leaf node or a cached copy of data $y_k$ at level $i$ for $i=1,\cdots, h(k).$  W.l.o.g., we consider node $v_{k, i}$  at level $i$.  If data $y_k$ is not cached from $v_{k, i}$ up to the sink node (level $0)$, i.e., $b_{k, j}=0$ for $j=0, \cdots, i,$  the cost is {incurred} due to  receiving, transmitting and compressing the data $(R_k-1)$ times, which is captured by the first term in Equation~\eqref{eq:cachretrievecost}, the second term is $0$.  Otherwise, the $(R_k-1)$ requests are served by the cached copy at $v_{k, i}$,  the corresponding caching and transmission cost serving from $v_{k, i}$ are captured by the second term in Equation~\eqref{eq:cachretrievecost}, and the corresponding reception, transmission and compression cost from $v_{k, i-1}$ upto to sink node is captured by the first term.  Note that the first time cost of reception, transmission and compression the data from leaf node to $v_{k, i}$ is {already} captured by Equation~\eqref{eq:servingcost}.  }

%\red{Equation~\eqref{eq:cachretrievecost} captures the energy cost of satisfying the remaining $(R_k-1)$ request either by using source leaf node or  a cached copy of the data $y_k$ at height $i$. }\red{The first term in Equation~\ref{eq:cachretrievecost} is the cost of receiving, transmitting and compressing the data $(R_k-1)$ times if the data is not cached up in the tree captured by the term ($1-\sum_{j=0}^{i}b_{k,j}$). The second term in Equation~\ref{eq:cachretrievecost} is the cost of caching data and transmitting it $(R_k-1)$ using the cached copy. Once a node caches data, it will not pull any data from the nodes lower in the tree and total energy cost will consist of the caching and transmission cost at node $v_{k,i}$, and the transmission, reception and compression cost at nodes higher in the tree. We give an example below to help understand the formulation }. %The second term captures the energy cost for storage and transmission by node $v_{k,i}$. \red{This can be interpreted as deciding the reduction rate  as well as whether the data should be cached or not at node $v_{k,i}$. } 
%{\bf[Still very confusing here for the physical meaning of equation~(\ref{eq:cachretrievecost}).]}

{We present a simple but illustrative example to explain the above equations.}
\begin{example}\label{exm1}
	{We consider a network with one leaf node and one sink node, i.e., $k=1$ and $h(k)=1.$ }
	%\red{We consider $k=1$ and $h(k)=1$ in~(\ref{eq:optimization}), i.e., one leaf node and one sink node. Then~(\ref{eq:servingcost}) and~(\ref{eq:cachretrievecost}) reduce to }
	{Then the cost in Equation~(\ref{eq:servingcost}) becomes
		%\begin{small}
		%\begin{align*}
		$E_1^C=y_1f(\delta_{1, 0})\delta_{1,1}+y_1f(\delta_{1,1}),$
		%\end{align*}
		%\end{small}
		where the first and second terms capture the reception, transmission and compression cost for data $y_1$ at sink node and the leaf node, respectively. }
	
	{The cost in Equation~(\ref{eq:cachretrievecost}) is $E_1^R=$
		\begin{small}
			\begin{align*}
			&\underbrace{y_1(R_1-1)\left[f(\delta_{1, 0})\delta_{1,1}(1-b_{1,0})+\delta_{1,0} \delta_{1,1}b_{1,0}\left(\frac{w_{ca}T}{R_1-1}+\varepsilon_{1T}\right)\right]}_{\text{Term $1$}}\nonumber\displaybreak[0]\\
			&{+\underbrace{y_1(R_1-1)\left[f(\delta_{1, 1})(1-b_{1,0}-b_{1,1})+ \delta_{1,1}b_{1,1}\left(\frac{w_{ca}T}{R_1-1}+\varepsilon_{1T}\right)\right]}_{\text{Term $2$}}},
			\end{align*}
		\end{small}
		where $\text{Term $1$}$ and $\text{Term $2$}$ capture the costs at sink node and leaf node, respectively.   To be more specific, there are three cases: (i) data $y_1$ is cached at sink node $0$, i.e., $b_{1, 0}=1$ and $b_{1, 1}=0$ (since we only cache one copy);  (ii) data $y_1$ is cached at leaf node $1$, i.e., $b_{1, 0}=0$ and $b_{1, 1}=1$; and (iii) data $y_1$ is not cached, i.e., $b_{1, 0}=b_{1, 1}=0$.  We consider these three cases in the following. }
	
	{Case (i), i.e., $b_{1, 0}=1$ and $b_{1, 1}=0$,  $\text{Term $2$}$ becomes $0$ and $\text{Term $1$}$ reduces to $y_1(R_1-1)\delta_{1,0} \delta_{1,1}b_{1,0}(\frac{w_{ca}T}{R_1-1}$\newline$ +\varepsilon_{1T})$ since all the $(R_1-1)$ requests are served from sink node. {This indicates that the total energy cost is due to caching the data for time period $T$ and transmitting it $(R_k-1)$ times from the sink node to users that request it. }} 
	
	{Case (ii), i.e., $b_{1, 0}=0$ and $b_{1, 1}=1$,  $\text{Term $1$}$ becomes $y_1(R_1-1)f(\delta_{1, 0})\delta_{1,1}$, which captures the reception, transmission and compression costs at sink node $0$ for serving the $(R_1-1)$ requests.  $\text{Term $2$}$ becomes $y_1(R_1-1) \delta_{1,1}b_{1,1}\left(\frac{w_{ca}T}{R_1-1}+\varepsilon_{1T}\right)$, which captures the cost {of caching data at the leaf node and transmitting the data $(R_k-1)$ times from the cached copy to the sink node} .  The sum of them is the total cost to serve $(R_1-1)$ requests. }
	
	{Case (iii), i.e., $b_{1, 0}=b_{1, 1}=0$, $E_1^R=y_1(R_1-1)f(\delta_{1, 0})\delta_{1,1}+y_1(R_1-1)f(\delta_{1, 1})$, which captures the reception, transmission and compression costs at sink node $0$ and leaf node $1$ for serving the $(R_1-1)$ requests since there is no cached copy in the network.  }

	%\red{$E_1^C$ is the total energy consumed at sink node (first term in $E_1^C$) and leaf node (second term in $E_1^C$) in receiving, transmitting, and compressing data $y_1$ for the first request received. The total data received at sink node is a fraction of $y_1$, i.e., reduced by a factor of $\delta_{1,1}$. Term 1 in $E_1^R$ consists of transmitting, receiving and compressing the data $(R_1-1)$ times  if the data is not cached at the sink node i.e., if $b_{1,0}$ is $0$. Term 2 in $E_1^R$ consists of transmitting, receiving and compressing the data $(R_1-1)$ times if the data is not cached at either leaf node or sink node i.e., if $b_{1,1}$ or $b_{1,0}$ is $0$. If data is cached either at leaf node or root node i.e., if $b_{1,1}$ or $b_{1,0}$ is $1$, then the node with cached copy will incur caching cost and $(R_1-1)$ times transmission cost.   } 
\end{example}

The total energy consumed in the network is $E^{\text{total}}$,
\begin{align}
E^{\text{total}}(\boldsymbol{\delta},\boldsymbol{b})\triangleq {\sum_{k \in \mathcal{K}}}\bigg(E^{\text{C}}_k+E^{\text{R}}_k\bigg),
\label{eq:total}
\end{align}
where $\boldsymbol{\delta}=\{\delta_{k, i}, \forall k \in \mathcal{K}, i=0,\cdots, h(k)\}$ and $\boldsymbol{b}=\{b_{k,i}, \forall k \in \mathcal{K}, i=0,\cdots, h(k)\}$. 
Our objective is to minimize the total energy consumption of the network with a QoI constraint for end users by choosing the compression ratio vector $\boldsymbol \delta$ and caching decision vector $\boldsymbol b$ in the network $G.$ Therefore, the optimization problem is,
\begin{subequations}\label{eq:optimization}
	\begin{align}\small
	\min_{\boldsymbol \delta, \boldsymbol b}\quad &E^{\text{total}}(\boldsymbol{\delta},\boldsymbol{b}) %\triangleq \red{\sum_{k \in \mathcal{K}}}E^{\text{C}}+E^{\text{R}}
	\\
	\text{s.t.} \quad&\sum_{k\in \mathcal{K}}y_{k} \prod_{i=0}^{h(k)}\delta_{k,i}\geq\gamma, \displaybreak[0] \label{con:1}\\
	&\sum_{k \in C_v} b_{k, {h(v)}} y_{k} \prod_{j=h(k)}^{h({v})}\delta_{k,j}\leq S_v, \forall \; v \in V,\label{con:2}\\
	%& \sum_{j=i}^{h(k)}a_{k,j}=0, \forall b_{k,i-1}=1, k\in\mathcal{K}, i-1=0, \cdots, h(k)\nonumber\displaybreak[2]\\
	& \sum_{i=0}^{h(k)}b_{k,i} \leq 1,\forall k\in\mathcal{K}, \label{con:3}\\ \displaybreak[3]
	& {0< \delta_{k, i}\leq 1, \forall k \in \mathcal{K},\; i=0, \cdots, h(k),}\label{con:4}\\
	&  b_{k,i}\in\{0,1\}, \forall  k\in\mathcal{K}, i=0, \cdots, h(k), \label{con:5} \displaybreak[1]
	%& f_{k,i}(\delta_{k,i})=\varepsilon_{kT}\delta_{k,i}+\varepsilon_{kC}(l_i(\delta_i)), \nonumber\\
	%&\qquad\qquad\forall b_{k,i}=1, k\in\mathcal{K},i=0, \cdots, h(k).
	\end{align}
	%\vspace{-0.05in}
\end{subequations}
\begin{comment}
\begin{align}\small
\min_{\boldsymbol \delta, \boldsymbol b}\quad &E^{\text{total}}(\boldsymbol{\delta},\boldsymbol{b})%\triangleq \red{\sum_{k \in \mathcal{K}}}E^{\text{C}}+E^{\text{R}}
\nonumber \\
\text{s.t.} \quad&\sum_{k\in \mathcal{K}}y_{k} \prod_{i=0}^{h(k)}\delta_{k,i}\geq\gamma, \nonumber\displaybreak[0]\\
&  b_{k,i}\in\{0,1\}, \forall  k\in\mathcal{K}, i=0, \cdots, h(k),  \nonumber\displaybreak[1]\\
&\sum_{k \in C_v} b_{k, {h(v)}} y_{k} \prod_{j=h(k)}^{h({v})}\delta_{k,j}\leq S_v, \forall \; v \in V,\nonumber\\
%& \sum_{j=i}^{h(k)}a_{k,j}=0, \forall b_{k,i-1}=1, k\in\mathcal{K}, i-1=0, \cdots, h(k)\nonumber\displaybreak[2]\\
& \sum_{i=0}^{h(k)}b_{k,i} \leq 1,\forall k\in\mathcal{K}, \displaybreak[3]
%& f_{k,i}(\delta_{k,i})=\varepsilon_{kT}\delta_{k,i}+\varepsilon_{kC}(l_i(\delta_i)), \nonumber\\
%&\qquad\qquad\forall b_{k,i}=1, k\in\mathcal{K},i=0, \cdots, h(k).
\label{eq:optimization}
\end{align}
\end{comment}
{where $h(v)$ is the depth of node $v$ in the tree. }

The first constraint is the QoI constraint, i.e., the total data available at the sink node \cite{nazemi2016qoi}. { The second constraint indicates that our decision (caching) variable  $b_{k,i}$ is binary. }
The {third} constraint is on {total amount of data that} can be cached at each node. 
The {fourth} constraint is that at most one copy of the generated data should be cached on the path between the leaf node and the sink node.

{The optimization problem in \eqref{eq:optimization} is a non-convex MINLP problem with   $M$ continuous variables, the $\delta_{k,i}$'s and $M$ binary variables, the $b_{k,i}$'s where, $M$ = $\sum_{k \in \mathcal{K}} h(k)$. }

\subsection{Properties}
{We first analyze the complexity of the problem given in \eqref{eq:optimization} and show that it is NP-hard.}

\begin{theorem}
	{The optimization problem in~(\ref{eq:optimization}) is NP-hard.}
\end{theorem}

\begin{proof}
We prove the hardness by a reduction of any given \emph{ $0-1$ knapsack problem} (KP) to a corresponding instance of \eqref{eq:optimization}. The KP is given as follows:
%	{We prove the hardness by a reduction from the \emph{$0-1$ Knapsack problem} (KP). Given a set of items numbered from $1$ up to $n,$ each with a weight $w_i$ and a value $v_i,$ along with a maximum weight capacity $W$.  Let $x_i$ represent the number of instances of item $i$ to be include in the knapsack, which is restricted to a maximum non-negative integer value $one$.  Then the KP is defined as: }
	\begin{align}\label{eq:bkp}
{	\max}\quad&{\sum_{k=1}^{|\mathcal{K}|} v_k x_k}\nonumber\\
	\text{{s.t.}}\quad&{\sum_{k=1}^{|\mathcal{K}|} w_k x_k\leq W},\nonumber\\
	& {x_k \in \{0,1\},} \quad \forall k \in K. %0\leq x_i\leq c.
	\vspace{-0.12in}
	\end{align}
	where $|\mathcal{K}|$ is the number of items, $x_k$ is a boolean variable that shows whether item $k$ is stored in the knapsack or not, $w_k$ is the weight of the item, $v_k$ is the value of the item and $W$ is the knapsack capacity.
	\par In the corresponding instance of \eqref{eq:optimization}, the processing tree has two levels; that is, the tree has $|\mathcal{K}|$ number of leaf nodes, which can generate data, and the sink node where only the sink node can choose to cache data generated by the leaf nodes. Furthermore, the threshold $\gamma$ is set to $\gamma=\sum_{k \in \mathcal{K}}y_k$ that enforces all the $\delta$'s to be one so that the constraint in \eqref{con:1} can also be satisfied and thus be removed. We also set the cache size at the sink node to be $W$ and the caching decisions of $b_k,\; \forall k\in\{1,\cdots,|\mathcal{K}|\}$ and the data sizes from all leaf node $y_k,\; \forall k\in\{1,\cdots,|\mathcal{K}|\}$ are one-to-one correspondent with $x_k,\; \forall k\in\{1,\cdots,|\mathcal{K}|\}$ and $w_k,\; \forall k\in\{1,\cdots,|\mathcal{K}|\}$ of the knapsack problem, respectively.   The sum $E_k^C$ in \eqref{eq:total} becomes a constant as all $\delta$'s have been set to $1$. As for the $E_k^R$ in \eqref{eq:cachretrievecost} can be reduced to:
	
\begin{align}\label{eq:eqer}
E^{\text{R}}_k&= \sum_{i=0}^{1}\Bigg\{y_{k}(R_k-1)f(1)\Bigg \} -b_{k}\Bigg(\sum_{i=0}^{1}\Bigg\{y_{k}(R_k-1)f(1)\Bigg \}- y_{k}(R_k-1) \left(\frac{w_{ca}T}{R_k-1}+\varepsilon_{kT}\right) \Bigg)
\end{align}	
Note that the first term in the above is a constant and can be removed from the objective function.  By setting the caching cost in the instance of \eqref{eq:optimization} to zero,  (7) becomes:

\begin{align}\label{eq:eqer2}
E^{\text{R}}_k&= -b_{k}\Bigg(y_{k}(R_k-1)\Bigg\{\sum_{i=0}^{1}f(1) -\varepsilon_{kT}\Bigg \}\Bigg)
\end{align}	
For the given value of $v_k$ in the knapsack problem, we choose $y_k$, $R_k$, $\varepsilon_{kT}$ so that Equation \eqref{eq:eqer2} is given by:
\begin{align}\label{eq:eqer3}
E^{\text{R}}_k&= -b_{k}c_k
\end{align}	 
where $c_k = y_{k}(R_k-1)\Big\{\sum_{i=0}^{1}f(1) -\varepsilon_{kT}\Big \} \geq 0$  as $f(1)\geq \varepsilon_{kT}$. As a result, solving the corresponding instance of \eqref{eq:optimization} provides the solution to the knapsack problem.  As the latter is NP-hard, so is \eqref{eq:optimization}.
\end{proof}

\begin{rem}\label{thm:increasingR}
	The objective function $E^{\text{total}}$ defined in~\eqref{eq:optimization} is monotonically increasing in the number of  requests $R_k$ for all $k\in\mathcal{K}$ provided that $\boldsymbol \delta$ and $\boldsymbol b$ are fixed.
	%provided that the other conditions hold fixed in the network.
	%The objective function in \eqref{eq:optimization}  is monotonically increasing in the domain $\mathbb{R}^+$ with the increase in the arrival of Requests $R_k$ for data of leaf node $k$ while other parameters are held fixed. 
\end{rem}
{Notice that \eqref{eq:servingcost} is independent of $R_k$ and \eqref{eq:cachretrievecost} is linear in $R_k$,  and its multipliers are positive. Hence, for any fixed $\boldsymbol{b}$ and $\boldsymbol{\delta}$, \eqref{eq:total} increases monotonically with $R_k$.}

\begin{rem}\label{thm:rootcaching}
	Given a fixed network scenario, if we increase the number of requests $R_k$ for the data generated by leaf node $k,$ then these data will be cached closer to the sink node or at the sink node, if there exists enough cache capacity, to reduce the overall energy consumption. 
	%With increase in the number of Requests $R_k$ (while all other parameters are held fixed), the data will be cached at root node (provided that the root node has space available) or closer to root node (if root node is full) to reduce the overall energy consumption
\end{rem}
{For fixed $\boldsymbol{\delta}$, observe from \eqref{eq:cachretrievecost} that energy consumption decreases if the cache is moved closer to the root as the nodes deep in the tree do not need to retransmit. }

\subsection{Relaxation of Assumptions}\label{sec:relaxation}
{In our model, we make several assumptions for the sake of simplicity. In the following, we discuss the relaxation of these assumptions.}

{While we assume that the network is structured as a tree, this assumption can be easily relaxed as long as there exists a simple fixed path from each leaf node to the sink node. The tree structure represents a simple topology that captures the key parameters in the optimization formulation without the complexity introduced by a general network topology. 
	Furthermore, for simplicity, we assume that all parameters across the nodes are identical, which is not necessary as seen from the cost function.  We also assume that only leaf nodes generate data. However, our model can be extended to allow intermediate nodes to generate data at the cost of added complexity. %Finally, rather than having a constant $R_k$, we can generalize our approach to the case where $R_k$ are drawn from a distribution such as the Zipf distribution \cite{choi2012network}.
}

\section{Variant of Spatial Branch-and-Bound Algorithm}\label{sec:sBNB}
In this section, we present a variant of the Spatial Brand-and-Bound algorithm (V-SBB).  Instead of solving the MINLP problem~(\ref{eq:optimization}) directly, we use V-SBB to solve a \emph{standard form} of the original MINLP.   We first introduce the \emph{Symbolic Reformulation}\cite{smith1996global} method that reformulates the MINLP~(\ref{eq:optimization}) into a standard form needed by V-SBB.
\begin{definition}
%\red{The equivalent of any standard MINLP problem can be written as}
A MINLP problem is said to be in a \emph{standard form} if it can be written as 
\begin{align}
\min_{\boldsymbol w}\quad &\boldsymbol w_{\text{obj}}\nonumber \displaybreak[0]\\
\text{s.t.} \quad&\mathbf{Aw=b}, \nonumber\displaybreak[1]\\
&\mathbf{w}^l\leq \mathbf{w}\leq \mathbf{w}^U,  \nonumber\displaybreak[2]\\
& \mathbf{w}_k \equiv \mathbf{w}_i \mathbf{w}_j,\quad \; \forall(i,j,k) \;\in \; \mathcal{T}_{\text{bt}},  \nonumber\displaybreak[3]\\
& \mathbf{w}_k \equiv \mathbf{w}_i/\mathbf{w}_j,\quad  \; \forall(i,j,k) \;\in \; \mathcal{T}_{\text{lft}},
\label{eq:genericre}
\end{align}
where the vector of variables $\boldsymbol w$ consists of continuous and discrete variables in the original MINLP.  The sets $\tau_{\text{bt}}$ and $\tau_{\text{lft}}$ contain all relationships that arise in the reformulation. $\boldsymbol A$ and $\boldsymbol b$ are a matrix and a vector of real coefficients, respectively.  The index $\text{obj}$ denotes the position of a single variable corresponding to the objective function value within the vector $\boldsymbol w.$ 
%as well as slack variables that have been introduced to convert the inequality constraints \red{\bf [comment: We are not using any slack variables as Leo in his phd thesis pointed out that Smith's method can work without any need for slack variables]} in the original MINLP to equalities and any auxiliary continuous variables introduced during the reformulation. The sets $\tau_{\text{bt}}$ and $\tau_{\text{lft}}$ contain all such relationships that arise in the reformulation. $A$ and $b$ are a matrix and a vector of real coefficients, respectively. \red{The index $\text{obj}$ denotes the position of a single variable corresponding to the objective function value within the vector $\boldsymbol w.$}
\end{definition}

%We first show that the MINLP problem~(\ref{eq:optimization}) can be converted to a \emph{standard form}.
%\newpage
\begin{theorem}
	The non-convex MINLP problem~(\ref{eq:optimization}) can be  transformed into a standard form.
\end{theorem}
Due to space constraints, we relegate detailed reformulations (see Appendix \ref{sec:appendixc} for details of symbolic reformulation) and standard form of~(\ref{eq:optimization}) to Appendix \ref{sec:appendixd}. 

Here, we give an example to illustrate the above reformulation process.

\begin{example}
	\normalfont
	{Consider the same network in Example~\ref{exm1}, 
		%We consider $k=1$ and $h(k)=1$ in~(\ref{eq:optimization}), i.e., one leaf node and one sink node. Then~(\ref{eq:servingcost}) and~(\ref{eq:cachretrievecost}) reduce to 
		%\begin{small}
		%\begin{align}
		%	&E_1^C=y_1f(\delta_{1, 0})\delta_{1,1}+y_1f(\delta_{1,1}), \nonumber\\
		%&E_1^R=y_1(R_1-1)\left[f(\delta_{1, 0})\delta_{1,1}+\delta_{1,0} \delta_{1,1}b_{1,0}(\frac{w_{ca}T}{(R_1-1)}+\varepsilon_{1T})\right]\nonumber\\
		%&+y_1(R_1-1)\left[f(\delta_{1, 1})(1-b_{1,0})+ \delta_{1,1}b_{1,1}(\frac{w_{ca}T}{(R_1-1)}+\varepsilon_{1T})\right],
		%\end{align}
		%\end{small}
		the non-convex MINLP problem becomes }
	\begin{align}\label{eq:opt-example2nodes}
	\min_{\boldsymbol \delta, \boldsymbol b}\quad &E^{\text{total}}(\boldsymbol{\delta},\boldsymbol{b})=E_1^C+E_1^R\nonumber\\
	\text{s.t.} \quad& y_1\delta_{1,0}\delta_{1,1}\geq\gamma, \nonumber\\
	&b_{1,0}, b_{1,1}\in\{0, 1\}, \nonumber\\
	&b_{1,0}y_1\delta_{1,0}\delta_{1,1}\leq S_0,\nonumber\\
	&b_{1,1}y_1\delta_{1,1}\leq S_1,\nonumber\\
	&b_{1,0}+b_{1,1}\leq 1.
	\end{align}
	$\delta_{1, 0} \delta_{1, 1}$ is a bilinear term.  Based on symbolic reformulation rules, a new bilinear auxiliary variable $w_{1,0}^{\text{bt}}$ needs to be added.   
	The first constraint in \eqref{eq:opt-example2nodes} is then transformed into $ y_{1} w_{1,0}^{\text{bt}} \geq \gamma,$ which is linear in auxiliary variable $w_{1,0}^{\text{bt}}$. 
	Similarly, we add $w_{1,0}^{\text{lft}}$ for linear-fractional term $\delta_{1, 1}/\delta_{1, 0}$ that appears in $f(\cdot).$ 
	$b_{1,0}\delta_{1, 0}\delta_{1, 1}$ in the third constraint of \eqref{eq:opt-example2nodes} is a tri-linear term.  Since $\delta_{1, 0}\delta_{1, 1}$ is replaced by $w_{1,0}^{\text{bt}}$, we obtain a bilinear term $b_{1,0}w_{1,0}^{\text{bt}}$.  Again,  based on symbolic reformulation rules, $b_{1,0}w_{1,0}^{\text{bt}}$ is replaced by a new auxiliary variable $\overline{w}_{1,0}^{\text{bt}}$. 
	Similarly we add new auxiliary variables $\tilde{w}_{1,1}^{\text{b}}, \tilde{w}_{1,0}^{\text{bt}},$, $\tilde{w}_{1,0}^{\text{lft}}$ and {$\tilde{w}_{1,1}^{\text{lft}}$}.  The objective function in~(\ref{eq:opt-example2nodes}) can be then expressed as a function of these new auxiliary variables.  Therefore, the standard form of \eqref{eq:opt-example2nodes} is 
%	\begin{small}
		\begin{align}
		\min_{\boldsymbol \delta, \boldsymbol b}\quad & w_{\text{obj}}
		\nonumber\\
		\text{s.t.} \quad& y_{1} w_{1,0}^{\text{bt}} \geq \gamma, \nonumber\\
		&  b_{1,0}, b_{1, 1}\in\{0,1\},  \nonumber\\
		& y_{1} \overline{w}_{1,0}^{\text{bt}}\leq S_0, \nonumber\\
		& y_{1} \tilde{w}_{1,1}^{\text{bt}}\leq S_1, \nonumber\\
		& b_{1,0}+b_{1, 1}\leq 1, \nonumber\\
		& w_{1,0}^{\text{bt}} = \delta_{1, 1}\times \delta_{1,0}, \nonumber\\
		& w_{1,0}^{\text{lft}} = \delta_{1, 1}/\delta_{1,0}, \nonumber\displaybreak[0]\\
		& \overline{w}_{1,0}^{\text{bt}} = b_{1,0} \times w_{1,0}^b, \nonumber\\
		& \tilde{w}_{1,1}^{\text{bt}} = b_{1, 1} \times \delta_{1, 1}, \nonumber\\
		&  \tilde{w}_{1,0}^{\text{bt}}= \delta_{1, 1}\times b_{1,0},\nonumber \\
		&  \tilde{w}_{1,0}^{\text{lft}}=  b_{1,0}/\delta_{1, 1},\nonumber\\
		&  \overline{w}_{1,0}^{\text{lft}}=  b_{1,0}w_{1,0}^{\text{lft}},\nonumber\\
		&  \tilde{w}_{1,1}^{\text{lft}}=  b_{1,1}/\delta_{1, 1},\nonumber\\
		& w_{\text{obj}} = y_1 \varepsilon_{1R} \delta_{1,1}+\varepsilon_{1T}y_1 w_{1,0}^{\text{bt}}+y_1 \varepsilon_{1C} w_{1,0}^{\text{lft}}-y_1 \varepsilon_{1C} \delta_{1, 1}+y_1 \varepsilon_{1R} \varepsilon_{1T} y_1 \delta_{1, 1} +y_1 \varepsilon_{1C}/\delta_{1,1}-y_1 \varepsilon_{1C}+ \nonumber\\
		& y_1(R_1-1)\bigg(\varepsilon_{1R}\delta_{1, 1}+\varepsilon_{1T}w_{1,0}^{\text{bt}}-\varepsilon_{1C}\delta_{1, 1}+\varepsilon_{1C}w_{1,0}^{\text{lft}}-\varepsilon_{1R}\tilde{w}_{1,0}^{\text{bt}}-\varepsilon_{1T}\overline{w}_{10}^{\text{bt}}+\varepsilon_{1C}\tilde{w}_{1,0}^{\text{bt}}-\varepsilon_{1C}\overline{w}_{1,0}^{\text{lft}}\bigg)+\nonumber \\
		& y_1(R_1-1)\varepsilon_{1T} + y_1 w_{ca}T\overline{w}_{1,0}^{\text{bt}}+y_1(R_1-1)\bigg(\varepsilon_{1R}- \varepsilon_{1C}+ \varepsilon_{1T}\delta_{1, 1}+\varepsilon_{1C}/\delta_{1, 1}-\varepsilon_{1R}b_{1,0}-\varepsilon_{1C}b_{1,0}-\nonumber \\
		&\varepsilon_{1T}\tilde{w}_{1,0}^{\text{bt}}-\varepsilon_{1C}\tilde{w}_{1,0}^{\text{lft}}-\varepsilon_{1R}b_{1,1}+\varepsilon_{1C}b_{1,1}-\varepsilon_{1T}\tilde{w}_{1,1}^{\text{bt}}-\varepsilon_{1C}\tilde{w}_{1,1}^{\text{lft}}\bigg).
			\label{eq:example2nodes}
		\end{align}
%	\end{small}
\end{example}
Through this reformulation, the non-convex and non-linear terms in the original problem are transformed into bilinear and linear fractional terms, which can be easily used to compute the lower bound of each region in V-SBB, which are discussed in details later.  This is the reason V-SBB requires reformulating the original problem into a standard form.

\begin{theorem}
	%Reformulated problem~(\ref{eq:exactproblem}) and the original MINLP \eqref{eq:optimization} are equivalent.}
	Reformulated problem and the original MINLP are equivalent.
\end{theorem}
Proof is available in Section $2$ (page $460$) \cite{smith1996global}.

Due to the reformulation, the number of variables in the reformulated problem is larger than in the original MINLP. In the following, we show that the number of auxiliary variables that arise from symbolic reformulation is bounded. 
\begin{rem}
	\label{thm:numberofvariables}
	The number of auxiliary variables in the symbolic reformulation is $O(n^2),$ where $n=2M$ is the number of variables in the original  formulation.
\end{rem}
%\begin{proof}
From \cite{smith1999symbolic}, a way to transform a general form optimization problem %\eqref{eq:generaloptimizationform} 
into a standard form \eqref{eq:genericre} is through basic arithmetic operations on original variables. To be more specific, any algebraic expression results from the basic operators including the five basic binary operators, i.e., addition, subtraction, multiplication, division and exponentiation, and the unary operators, i.e., logarithms etc.  Therefore, in order to construct a standard problem consisting of simple terms corresponding to these binary or unary operations, new variables need to be added corresponding to these operations.  From the symbolic reformulation process \cite{smith1999symbolic, smith1996optimal, liberti2004reformulation}, any added variable results from the basic operations between two (including possibly the same) original variables or added variables. Hence, based on the basic operations, there are at most $n^2$ combinations of these variables, given that there are $n$ variables in the original problem \eqref{eq:optimization}. Therefore, the number of added variables in the symbolic reformulation is bounded as $O(n^2).$
%\end{proof}
In the remainder of this section, we present the V-SBB to solve the equivalent problem. % defined in~(\ref{eq:exactproblem}).% which is equivalent to the original MINLP problem in \eqref{eq:optimization}. 
\begin{algorithm}
	\begin{algorithmic}[]
		%\DontPrintSemicolon % Some LaTeX compilers require you to use \dontprintsemicolon    instead
		%	\begin{comment}
		%    \textbf{Input:}\\
		%    \textbf{Output:}\\
		%    \textbf{Step 1}: Initialize $\phi^u:=\infty$ and $\mathcal{L}$ to a single domain  \\
		%    \textbf{Step 2}: Choose a subregion $\mathcal{R} \in \mathcal{L}$ \\
		%    \textbf{Step 3}: Obtain the lower bound $\phi^{\mathcal{R},l}$ \\
		%                if
		%	\end{comment}
		%	\State \textbf{Input:} 
		%	\State \textbf{Output:} $\phi^u$
		\State \textbf{Step 1}: Initialize $\phi^u:=\infty$ and $\mathcal{L}$ to a single domain  
		\State \textbf{Step 2}: Choose a subregion $\mathcal{R} \in \mathcal{L}$ using \emph{least lower bound rule}
		\If{ $\mathcal{L}$ = $\emptyset$} Go to Step 6
		%	\State Go to Step 6
		\EndIf
		\If{ for chosen region $\mathcal{R}$,  $\phi^{\mathcal{R},l}$ is infeasible or $\phi^{\mathcal{R},l}\geq \phi^u-\epsilon$} Go to Step 5
		%	\State Go to Step 5
		\EndIf
		%	\Else { Go to Step 3}
		%	\State		Go to step 3
		
		%	\State  \textbf{Step 3}: Obtain the lower bound $\phi^{\mathcal{R},l}$ 
		
		%	\Else {Go to Step 4:}
		%	\EndIf 
		\State \textbf{Step 3}: Obtain the upper bound $\phi^{\mathcal{R},u}$ 
		\If{upper bound cannot be obtained or if $\phi^{\mathcal{R},u}>\phi^u$} Go to Step 4
		%	\State Go to Step 4
		\Else{ $\phi^u:=$$\phi^{\mathcal{R},u}$} and, from the list $\mathcal{L}$, delete all subregions $\mathcal{S} \in \mathcal{L}$ such that $\phi^{\mathcal{S},l}\geq \phi^u-\epsilon$
		\EndIf
		\If{$\phi^{\mathcal{R},u}-\phi^{\mathcal{R},l}\leq \epsilon$} Go to Step 5
		%\State   Go to Step 5
		\EndIf
		\State \textbf{Step 4}: Partition  $\mathcal{R}$ into new subregions  $\mathcal{R}_{\text{right}}$ and  $\mathcal{R}_{\text{left}}$
		\State \textbf{Step 5}: Delete $\mathcal{R}$ from $\mathcal{L}$ and go to Step 2
		\State \textbf{Step 6}: Terminate Search
		\If{$\phi^u=\infty$}  Problem is infeasible
		%\State Problem is infeasible
		\Else { $\phi^u$ is  $\epsilon$-global optimal}
		\EndIf
	\end{algorithmic}
	\caption{Variant of Spatial Branch-and-Bound (V-SBB)}% (V-SBB)}%for  problem in \eqref{eq:exactproblem}}
	\label{algo:sbnb}
\end{algorithm}

\subsection{Our Variant of Spatial Branch-and-Bound}
The proposed spatial branch-and-bound method is a variant of the method proposed in \cite{smith1999symbolic} and is primarily tuned for solving our optimization problem \eqref{eq:exactproblem} that is also the solution of \eqref{eq:optimization}.  Our algorithm is  different from \cite{smith1999symbolic} because
	\begin{itemize}
		\item We do not use any bounds tightening steps as it does not always guarantee faster convergence \cite{liberti2004phdthesis} and in case of our problem slowed down the process. 
		\item   By eliminating the bounds tightening step, we do not need to calculate the lower bound $\phi^{\mathcal{R},l}$ again separately and utilize the lower bound obtained in Step 2 for the chosen region $\mathcal{R}$, hence reducing the computational complexity of the algorithm.
	\end{itemize}
	
 Algorithm \ref{algo:sbnb} provides an overview of the steps involved in spatial branch-and-bound algorithm. We describe some of the steps in Algorithm \ref{algo:sbnb} in detail below. 
\begin{comment}
\begin{algorithm}
	\begin{algorithmic}[1]
		\State \textbf{Step 1}: Initialize $\phi^u:=\infty$ and $\mathcal{L}$ to a single domain  
		\State \textbf{Step 2}: Choose a subregion $\mathcal{R} \in \mathcal{L}$ using \emph{least lower bound rule}
		\If{ $\mathcal{L}$ = $\emptyset$} 
		\State Go to Step 6
		\EndIf
		\If{ for chosen region $\mathcal{R}$,  $\phi^{\mathcal{R},l}$ is infeasible or $\phi^{\mathcal{R},l}\geq \phi^u-\varepsilon$}
		\State Go to Step 5
		\EndIf
		\State \textbf{Step 3}: Obtain the upper bound $\phi^{\mathcal{R},u}$ 
		\If{upper bound cannot be obtained or if $\phi^{\mathcal{R},u}>\phi^u$}
		\State Go to Step 4
		\Else{ $\phi^u:=$$\phi^{\mathcal{R},u}$} and, from the list $\mathcal{L}$, delete all subregions $\mathcal{S} \in \mathcal{L}$ such that $\phi^{\mathcal{S},l}\geq \phi^u-\epsilon$
		\EndIf
		\If{$\phi^{\mathcal{R},u}-\phi^{\mathcal{R},l}\leq \epsilon$}
		\State  Go to Step 5 
		\EndIf
		\State \textbf{Step 4}: Partition  $\mathcal{R}$ into new subregions  $\mathcal{R}_{right}$ and  $\mathcal{R}_{left}$
		\State \textbf{Step 5}: Delete $\mathcal{R}$ from $\mathcal{L}$ and go to Step 2
		\State \textbf{Step 6}: Terminate Search
		\If{$\phi^u=\infty$} 
		\State Problem is infeasible
		\Else { $\phi^u$ is the $\varepsilon$-global optimal}
		\EndIf
	\end{algorithmic}
	\caption{Variant Spatial Branch-and-Bound (V-SBB)\red{\bf Make the algorithms the same as 9 pages.}}% (V-SBB)}%for  problem in \eqref{eq:exactproblem}}
	\label{algo:sbnb}
\end{algorithm}
\end{comment}

\noindent{\textit{Step $2$:}} There are a number of approaches that can be used to choose a subregion $\mathcal{R}$ from $\mathcal{L}$ \cite{floudas2013deterministic}. Here we use the \emph{least lower bound rule}, i.e., we choose a subregion $\mathcal{R} \in \mathcal{L}$ that has the lowest lower bound among all the subregions, since it is a widely used and well researched method. The lower bound can be obtained by solving a convex relaxation of the problem in \eqref{eq:exactproblem}.  As our optimization problem in \eqref{eq:optimization} and \eqref{eq:exactproblem} contains only bilinear and linear fractional terms, we use McCormick  linear over-estimators and under-estimators \cite{mccormick1976computability} (see Appendix \ref{sec:appendixA}) to obtain a convex relaxation of all such terms. The resulting problem is then a Mixed Integer Linear Programming (MILP) problem that we solve using the \emph{SCIP} solver \cite{achterberg2009scip}. The SCIP solver is a faster and well known solver for MILP problems.  The subregion with lowest lower bound is then used as the region to explore for an optimum. The chosen regions' lower bound is used as $\phi^{\mathcal{R},l}$. If the convex relaxation is infeasible or if the obtained lower bound is higher than the existing upper bound $\phi^u$ of the problem, we \emph{fathom} or delete the current region by moving to step $5$.\\
%\subsubsection{Step 3}
%In step 3, the lower bound $\phi^{\mathcal{R},l}$ for the region $\mathcal{R} \in \mathcal{L}$ selected in Step 2 is to be calculated. As our Step 2 was based on the least lower bound procedure, we can use the lower bound value obtained in Step 2 as our $\phi^{\mathcal{R},l}$ hence reducing the computational cost of doing the same calculation all over again.% If any other approach is used in Step 2 to select region, then we need to calculate the lower bound for the region again by solving a convex relaxation of the problem \eqref{eq:exactproblem} in the given subregion. 
%If the convex relaxation is infeasible or if the obtained lower bound is higher than the existing upper bound $\phi^u$ of the problem, we \emph{fathom} or delete the current region by moving to step 6.
\noindent{\textit{Step $3$:}} In step $3$, we calculate the upper bound  $\phi^{\mathcal{R},u}$  for the subregion $\mathcal{R}$ chosen in Step $2$. This can be done in a number of ways (see \cite{smith1999symbolic}), here we use local MINLP solver such as \emph{Bonmin} \cite{bonami2008algorithmic} to obtain a local minimum for the subregion as it performed better in terms of time than using local non-linear programming optimization with fixed discrete values or added discreteness constraints in our simulation settings. If the upper bound for the region $\phi^{\mathcal{R},u}$ cannot be obtained or if it is greater than $\phi^{u}$ then we move to Step $4$ to further divide the region and search further for a better solution. Otherwise we set it as the current best solution $\phi^{u}$ and delete all the subregions whose lower bound is greater than the obtained upper bound since all such regions cannot contain the $\epsilon$-global optimal solution. If the difference between the upper and lower bound for the region is within the $\varepsilon$-tolerance, the current subregion need not to be searched further, then we delete the current subregion by going to step $5$, otherwise we move to step $4$ for further searching in the space. \\
\noindent{\textit{Step $4$:}} 
Step $4$ also known as the branching/partitioning step %is one of the most important steps of the algorithm as it
 helps in partitioning/dividing a region to further refine the search for solution. In branching step, we select a variable for branching/partitioning as well as the value of the variable at which the region is to be divided. There are a number of different rules and techniques that can be used for branching (see \cite{floudas2013deterministic} for detailed discussion). Here we use the variable selection and value selection rule specified in \cite{smith1996optimal}, since it has been found efficient for our problem \cite{smith1996optimal}. 

%\begin{itemize}
%	\item Variable Selection: 
%Selection of branching variable is fundamental to the faster convergence and better performance of any branch and bound algorithm. 
We branch on the variable that causes the maximum reduction in the feasibility gap between the solution of convex relaxation (solution of Step 2) and the exact problem. To do so, the approximation error for the bilinear and linear fractional terms in \eqref{eq:exactproblem} is calculated using \eqref{eq:aperrorbt} and \eqref{eq:aperrorlft} respectively where $\mathcal{S}_2$ means the value of the variable obtained in Step 2. The variable with the maximum approximation error of all is chosen as the branching variable as that tightens the gap between  the relaxation and the exact problem \cite{smith1996optimal}. This results in two candidate variables for branching i.e. $w_i$ and $w_j$. If one of the variables is discrete (binary in our case) and the other is continuous then choose the discrete variable since it will result only in finite number of branches. However, if both variables are of the same type (either binary or continuous), then the branching variable is chosen using \eqref{eq:variableselection} i.e. we choose the variable $w_b$ that has its value $w_b^\mathcal{R}$ closer to its range's midpoint. However,  we first need to obtain the branching value for the candidate variables $w_{bc}^\mathcal{r}$ (the value at which to branch).  $w_{bc}^\mathcal{r}$ should be between the upper and lower bounds of the variable in the region i.e. $w_{bc}^{\mathcal{r},l} < w_{bc}^\mathcal{r} < w_{bc}^{\mathcal{r},u}$. The rules for the choice of the branch point have been set in \cite{smith1996optimal}, however we restate them here for sake of completeness. 
\begin{itemize}
	\item Set $w_{bc}^\mathcal{r}$ to the value obtained in Step $2$, i.e., $w_{bc}^\mathcal{r}:=w_{bc}^{\mathcal{S}_2}.$
	\item If any feasible upper bound $\phi^u=\phi(w_{bc}^*)$ has been obtained and  $w_{bc}^{\mathcal{r},l} < w_{bc}^* < w_{bc}^{\mathcal{r},u}$, then  $w_{bc}^\mathcal{r}:=w_{bc}^*$ and stop the search for the value.
	\item If step $4$ provided an upper bound $\phi^{\mathcal{R},u}$ for the subregion $\mathcal{R}$, then  $w_{bc}^\mathcal{r}:=w_{bc}^\mathcal{R}$.
\end{itemize}
After obtaining the branch point value, we have all the parameters required for \eqref{eq:variableselection} and can then choose the variable for branching.
%The method to obtain the values $w_i^\mathcal{R}$ and $w_j^\mathcal{R}$  have been specified in \cite{smith1996optimal}, however we restate them in the appendix for completeness (\textbf{add the appendix}). 
	\begin{subequations}
	\begin{align}
	%\centering
	&E^{ijk}_{\text{bt}}=|w_k^{\mathcal{S}_2}-w_i^{\mathcal{S}_2}w_j^{\mathcal{S}_2}|\; \forall(i,j,k) \;\in \; \mathcal{T}_{\text{bt}} \displaybreak[0] \label{eq:aperrorbt}\\
	%\end{equation}
	%\begin{equation}
	%\centering
	&E^{ijk}_{\text{lft}}=\bigg|w_k^{\mathcal{S}_2}-\frac{w_i^{\mathcal{S}_2}}{w_j^{\mathcal{S}_2}}\bigg|\; \forall(i,j,k) \;\in \; \mathcal{T}_{\text{lft}} \label{eq:aperrorlft}
	\end{align}
	\end{subequations}
		\begin{align}
	w_{b}=\arg \; \min\;\Bigg\{\bigg|0.5-\bigg(\frac{w_i^\mathcal{r}-w_i^{\mathcal{R},l}}{w_i^{\mathcal{R},u}-w_i^{\mathcal{R},l}}\bigg)\bigg|, 
	\bigg|0.5-\bigg(\frac{w_j^\mathcal{r}-w_j^{\mathcal{R},l}}{w_j^{\mathcal{R},u}-w_j^{\mathcal{R},l}}\bigg)\bigg| \Bigg\}
	\label{eq:variableselection}
	\end{align}
%	\item Value of the Variable:
%\end{itemize}
%Once branching variable $w_b$ is selected, the next step is to determine the branch point 
\par We  partition the subregion $\mathcal{R}$ into $\mathcal{R}_{\text{right}}$ and $\mathcal{R}_{\text{left}}$ and add $\mathcal{R}_{\text{right}}$, $\mathcal{R}_{\text{left}}$ into our region list $\mathcal{L}$. Then we move to \emph{Step $5$} and delete the subregion $\mathcal{R}$ from the list $\mathcal{L}$.

\subsection{Convergence of Spatial Branch-and-Bound}
The spatial branch-and-bound method guarantees convergence to $\epsilon$-global optimality, which has been proven in \cite{liberti2004phdthesis}. However, for sake of completeness, we restate the proof in the Appendix \ref{sec:appendixB}.

\section{Evaluation}\label{sec:results}
\begin{figure}
	\centering
	\includegraphics[width=0.45\textwidth]{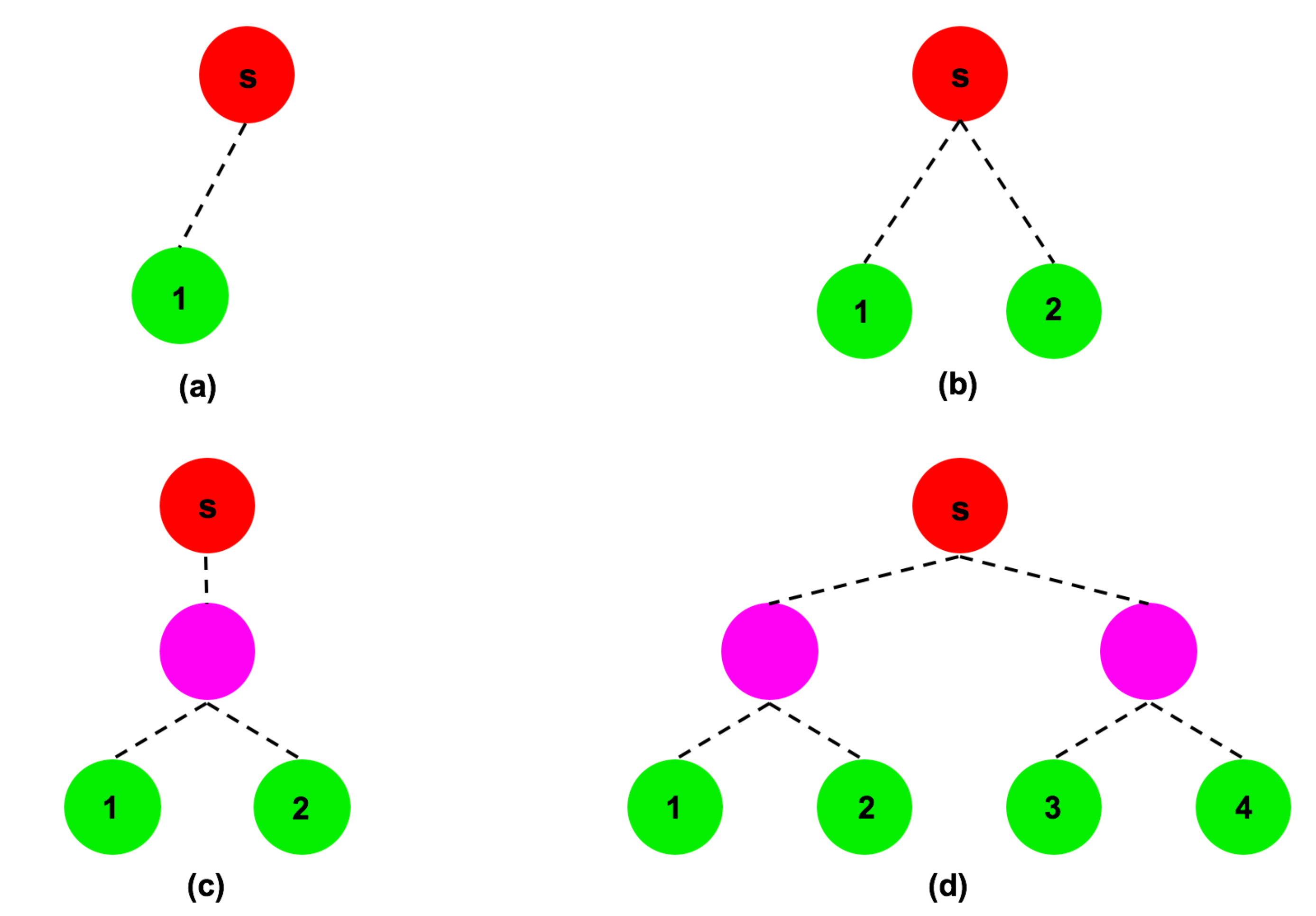}
	\vspace{-0.1in}
	\caption{Candidate network topologies used in the experiments: (a) one sink node and one leaf node; (b) one sink node and two leaf nodes; (c) one sink node, one intermediate node and two leaf nodes; and (d) one sink node, two intermediate nodes and four leaf nodes.}
	\protect\label{fig:networks}
	%\vspace{-0.25in}
\end{figure}
{We evaluate the performance of  our communication, compression and caching (C$3$) joint optimization framework through a series of experiments on several network topologies as shown in Figure~\ref{fig:networks}. Our goal is to analyze the performance of C$3$ and assess the improvement in energy efficiency that can be achieved by jointly considering C$3$ costs when compared with C$2$. While highlight the performance gain is valuable, characterizing the performance of C$3$ in different settings and parameters, and obtaining the optimal caching location and data compression rate is also of great significance. We also compare the performance of our  V-SBB algorithm with some other well-known solvers. }

 The highlights of the evaluation results are:
\begin{itemize}
		\item {Our C$3$ joint optimization framework improves energy efficiency by as much as $88\%$ compared to the C$2$ optimization over communication and computation, or communication and caching. This shows the significance of jointly considering C$3$ energy costs.   }
		\item {The improvement in energy efficiency with C$3$ framework increases with an increase in the number of requests and the network size. Furthermore, data of nodes that had largest number of requests $R_k$'s are cached at the sink node or closer to the sink node. }
	\item {While comparing different MINLP solvers, } V-SBB algorithm can obtain an $\epsilon$-global optimal solution in most situations. {We vary the network parameters and find that V-SBB is able to obtain a feasible solution in all settings. SCIP, Baron, Bonmin and Antigone are faster in obtaining solutions. However, they are either not able to obtain solutions in all the settings or they provide an objective value higher than our algorithm particularly for lower values of $\gamma$.  }
%	\item {}
	%the lowest value of the objective function in most situations, i.e., mostly close to the globally optimal solution in reasonable time. 
	%\item When Bonmin \cite{bonami2008algorithmic} can achieve a solution, it is faster. However, the solution obtained through Bonmin is not always comparable to that of V-SBB.  We observe that when higher compression is done (i.e., smaller value of $\gamma$), V-SBB provides a better solution (in terms of the value of objective function) than Bonmin.  More importantly, we find that Bonmin  cannot even produce feasible solutions in some cases although they exist.  NOMAD \cite{le2011algorithm} and GA \cite{deb2002fast} often produce  objective-function values much larger than  V-SBB.
\end{itemize}

\vspace{-0.2in}
\subsection{Methodology}
{Our primary goal is to highlight the improvement in energy efficiency that is achieved using the C$3$ framework when compared with C$2$.  We define the energy efficiency as:} 
\begin{align}\label{eq:efficiency}
\mathcal{E}=\frac{E^{\text{total}*}(\text{C}2)-E^{\text{total}*}(\text{C}3)}{E^{\text{total}*}(\text{C}2)}\times 100\%,
\end{align}
where $E^{\text{total}*}(\text{C}3)$ and $E^{\text{total}*}(\text{C}2)$ are the optimal energy costs under the C$3$ optimization framework in~\eqref{eq:optimization} and the C$2$ optimization, respectively.  $\mathcal{E}$ reflects the reduction of energy efficiency for the C$3$ over the C$2$ optimization. 
{While, the increase in energy efficiency using C$3$ framework is noteworthy, characterizing the magnitude of the improvement and the parameters that significantly impact the energy efficiency is important. Such characterization can help in identifying the \emph{operation regions} for the network and then accordingly devising heuristic algorithms for specific operation regions. We also compare the performance of V-SBB with other MINLP solvers and show that it performs comparably with other MINLP solvers for our C$3$ framework. }

\begin{table}[ht]
%	\vspace{-0.2in}
	\centering
	\caption{Characteristics of the solvers used in this paper}
	\vspace{-0.1in}
	\begin{tabular}{|l|p{13.0cm}|}
		%\begin{tabular}{|c|c{6.0cm}|}
		\hline
		\textbf{Solver} & \textbf{Characteristics}  \\ \hline
		\textbf{Bonmin} \cite{bonami2008algorithmic} &  A deterministic approach based on Branch-and-Cut method that solves relaxation problem with Interior Point Optimization tool (IPOPT), as well as mixed integer problem with Coin or Branch and Cut (CBC).  \\ \hline
		%A deterministic  MINLP solver based on branch-and-cut method that solves  the relaxed problem using Interior Point Optimization tool (IPOPT) and uses Coin or Branch and Cut (CBC) for the mixed integer problem                             \\ \hline
		\textbf{NOMAD} \cite{le2011algorithm}  & A stochastic approach based on Mesh Adaptive Direct Search Algorithm (MADS) that guarantees local optimality. It can be used to solve non-convex MINLP and has a relatively good performance.\\ \hline% for derivative free optimization.                      \\ \hline
		\textbf{GA} \cite{deb2002fast} &  A meta-heuristic stochastic approach that can be tuned to solve global optimization problems. We use Matlab \emph{Optimization Toolbox}'s implementation.      \\ \hline
			\textbf{{SCIP\cite{achterberg2009scip}}}  &  {One of the fastest, non-commercial, deterministic global optimization solver that uses branch-and-bound algorithm for solving MINLP problems.}     \\ \hline
			\textbf{{Baron\cite{tawarmalani2005polyhedral}}}  &  {A deterministic global solver for MINLP problems that relies on Branch and Cut approach for solving MINLP problems.}      \\ \hline
			\textbf{{Antigone\cite{misener2014antigone}}}  &  {A deterministic global solver for MINLP problems that relies on special structure of the problem and uses Branch and Cut approach to solve the problem.}      \\ \hline
	\end{tabular}
%	\vspace{-0.15in}
	\label{tab:Solvercharacteristics}	
\end{table}

%percentage of additional cost of C$2$ over C$3$, hence a positive value of $\mathcal{E}$ is preferred, i.e., C$3$ enhances the energy efficiency.  

%In practice, we prefer a highly stable algorithm to achieve an $\epsilon$-global optimal solution in relatively short time. 
%\item Effectiveness: Effectiveness of the algorithm can be defined as the ability of the algorithm to  provide a feasible solution. Effectiveness is one of the more important factors as obtaining a feasible solution is important from practical perspective. 
%\end{enumerate}
%\end{comment}

\noindent{\textbf{\textit{Setup:}}} 
We implement V-SBB in Matlab on a Core i$7$ $3.40$ GHz CPU with $16$ GB RAM.  The candidate MINLP solvers in this work include Bonmin, NOMAD and GA, which are implemented with Opti-Toolbox \cite{optitoolbox}. We summarize the characteristics of these solvers in Table~\ref{tab:Solvercharacteristics}.  Note that these solvers can be applied directly to solve the original optimization problem in \eqref{eq:optimization}, while our V-SBB solves the equivalent problem. % in \eqref{eq:exactproblem}.
The reformulations needed are executed by a Java based module and we derive the bounds on the auxiliary variables. %{given the derived bounds on the auxiliary variables}. 
We also relax the integer constraint in \eqref{eq:optimization} to obtain a non-linear programming problem, which is solved by IPOPT \cite{wachter06} and use it as a benchmark for comparison.  V-SBB terminates when $\epsilon$-optimality is obtained or a computation timer of {$400$} seconds expires.  We take $\epsilon=0.001$ in our study.  If the timer expires, the last feasible solution is taken as the best solution.  {For cases, where no solution is obtained within the specified timer, we increase the timer limit to $7200$ seconds.}
%\red{\bf We assume that the transmission cost at sink node is ignore if the data is cached at sink node since our formulation focuses on the cost inside the network.  Thus, once the data is cached at sink node, the energy cost for requesting that data becomes constant since there is no reception and transmission cost.} 
Our simulation parameters are provided in Table~\ref{tab:algcompareparameters}, which  are the typical values used in the literature \cite{nazemi2016qoi, heinzelman2000energy, ye2002energy}.

\begin{table}[ht]
%	\vspace{-0.2in}
	\centering
	\caption{Parameters used in simulations}%Simulation Parameters for Comparing Algorithms}
	%\begin{tabular}{|l|l||l|l|}
	\vspace{-0.1in}
	\begin{tabular}{|c|c||c|c|}
		\hline
		\textbf{Parameter} & \textbf{Value} &\textbf{Parameter} & \textbf{Value (Joules)} \\ \hline
		$y_k$	& 1000 &  $\varepsilon_{vR}$ &  50  $\times$ 10$^{-9}$             \\ \hline
		$R_k$	& 100    & $\varepsilon_{vT}$&   200  $\times$ 10$^{-9}$              \\ \hline
		$w_{ca}$	& 1.88 $\times$ 10$^{-6}$&  $\varepsilon_{cR}$ &  80  $\times$ 10$^{-9}$               \\ \hline
		$T$	& 10s   &$\gamma$ & [$1,\sum_{k \in \mathcal{K}}y_k$]               \\ \hline
	\end{tabular}
%	\vspace{-0.2in}
	\label{tab:algcompareparameters}
\end{table}

\begin{figure}
	\centering
	\vspace{-0.1in}
	\includegraphics[width=0.8\textwidth]{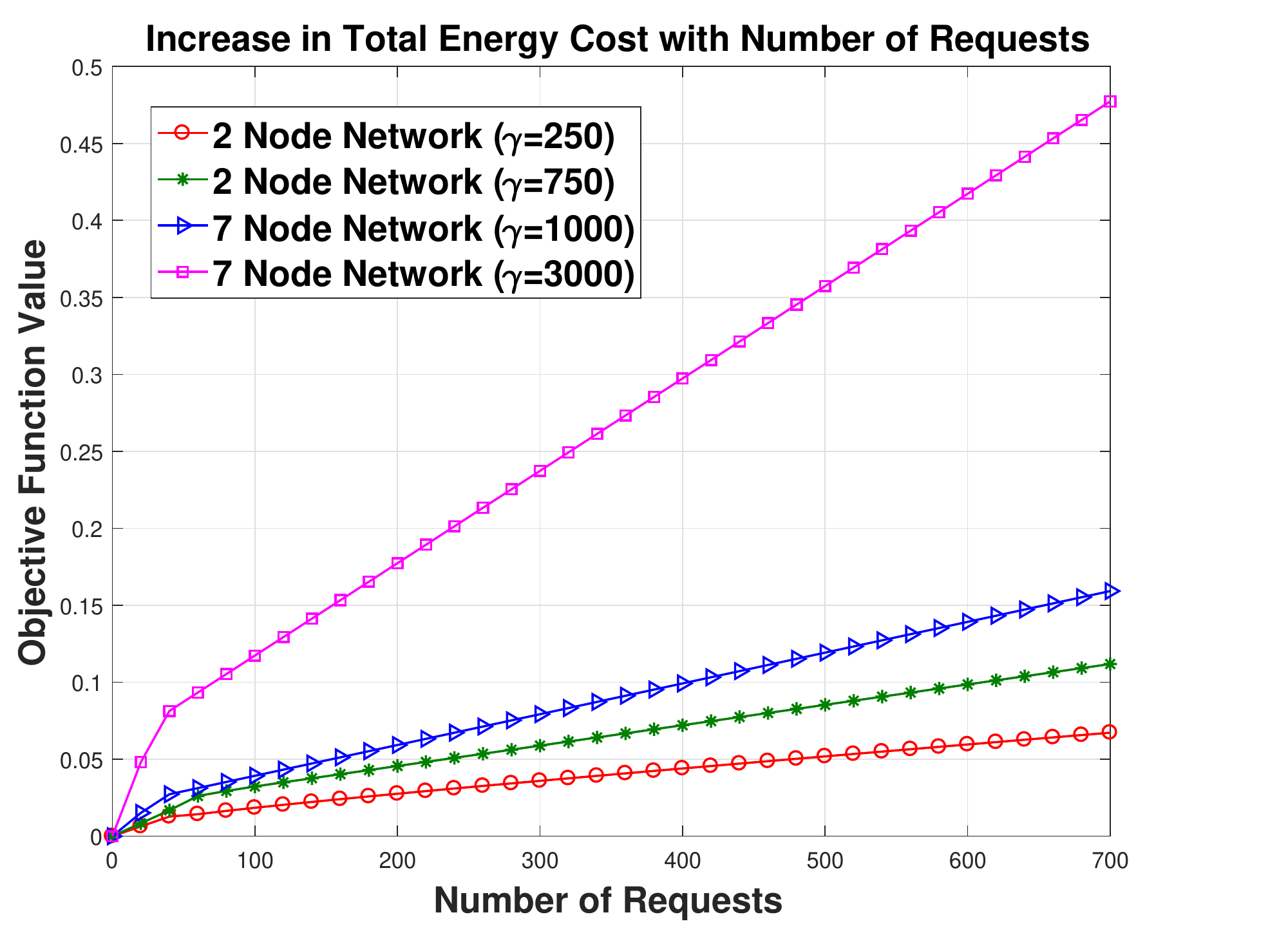}
	%\vspace{-0.3in}
	\caption{Total Energy Costs vs. Number of Requests.}
	\label{fig:mono2node}
%	\vspace{-0.3in}
\end{figure} 

\begin{table*}[]
	\centering
	\caption{The Best Solution to the Objective Function (Obj.) and Convergence time for two nodes network}
	\vspace{-0.1in}
	\begin{tabular}{|l|l|l|l|l|l|l|l|l|l|l|}
		\hline
		\multirow{2}{*}{\textbf{Solver}} & \multicolumn{2}{c|}{{$\gamma=1$}}               & \multicolumn{2}{c|}{{$\gamma= 250$}}      & \multicolumn{2}{c|}{{$\gamma=500$}}     & \multicolumn{2}{c|}{\textbf{$\gamma=750$}}     & \multicolumn{2}{c|}{\textbf{$\gamma=1000$}}    \\ \cline{2-11} 
		& \textbf{Obj.} & \textbf{Time} (s) & \textbf{Obj.} & \textbf{Time} (s) & \textbf{Obj.} & \textbf{Time} (s)& \textbf{Obj.} & \textbf{Time} (s) & \textbf{Obj.} & \textbf{Time} (s) \\ \hline
		\textbf{Bonmin}                  &  0.010&	0.076& 
		0.018&	0.07&	0.026&	0.071&	0.032&	0.077&	0.039&	0.102 \\ \hline
		\textbf{NOMAD}                   & 0.012&	1.036& 
		0.038&	0.739&	0.033&	0.640&	0.038&	0.203&	0.039&	0.263
		\\ \hline
		\textbf{GA}                      & 0.010&	0.286
		&	0.018&	2.817&	0.026&	7.670&	0.042	&11.020	&0.064	&3.330
		\\ \hline
		\textbf{V-SBB}                     & 0.010& 	18.231&
		0.018& 	17.389& 	0.026& 	12.278& 	0.032& 	7.327& 	0.039& 	19.437
		\\ \hline
		
			\textbf{{SCIP}}                     &{Inf} & {0.07}	&{0.0012}
		&{0.07} & { 0.005}	&{0.05} 	& {0.011}	&{0.087} & {0.039}	& 	{0.05}		\\ \hline
		
			\textbf{{Baron}}                   & {0.01 }& {0.91}	&{0.018}
		&{0.79} & {0.026}	& {0.77}	& {0.032}	&{0.87} & {0.039}	&{ 0.49}			\\ \hline
		
		\textbf{{Antigone}}                   & {0.01 }& {0.195}	&{0.018}
	&{0.18} & {0.026}	& {0.175}	& {0.032}	&{0.19} & {0.039}	&{0.2}			\\ \hline
		
	%	\textbf{Relaxed}                 & 0.010&	0.075
	%	&	0.018&	0.048	&0.026&	0.046&	0.032&	0.050&	0.039&	0.059
	%	\\ \hline
	\end{tabular}
	\label{tab:results2node}
%	\vspace{-0.15in}
\end{table*}

\begin{table*}[]
	\centering
	\caption{The Best Solution to the Objective Function (Obj.) and Convergence time for seven nodes network}
	\vspace{-0.1in}
	\begin{tabular}{|l|l|l|l|l|l|l|l|l|l|l|}
		\hline
		\multirow{2}{*}{\textbf{Solver}} & \multicolumn{2}{c|}{{$\gamma=1$}}            & \multicolumn{2}{c|}{\textbf{$\gamma=1000$}}      & \multicolumn{2}{c|}{\textbf{$\gamma=2000$}}     & \multicolumn{2}{c|}{\textbf{$\gamma=3000$}}     & \multicolumn{2}{c|}{\textbf{$\gamma=4000$}}    \\ \cline{2-11} 
		& \textbf{Obj.} & \textbf{Time} (s) &  \textbf{Obj.} &\textbf {Time} (s) & \textbf{Obj.} & \textbf{Time} (s) & \textbf{Obj.} & \textbf{Time} (s) & \textbf{Obj.} & \textbf{Time} (s) \\ \hline
		\textbf{Bonmin}                  & 0.0002&	0.214
		&	0.039&	0.164&	0.078&	0.593&	0.117&	0.167&	0.156	&0.212
		
		\\ \hline
		\textbf{NOMAD}                   & 0.004&	433.988& 
		0.121&	381.293&	0.108&	203.696&	0.158&	61.093&	0.181&	26.031
		\\ \hline
		\textbf{GA}                      & 0.043&	44.538	
		&0.096	&30.605	&0.164	&44.970	&0.226&	17.307	&0.303	&28.820
		\\ \hline
		\textbf{V-SBB}                     & 0.0001& 	1871.403&
		0.039& 	25.101& 	0.078& 	30.425& 	0.117& 	23.706& 	0.156	& 19.125
		\\ \hline
	
\textbf{{SCIP}}                     &{NC} & {5901.7}	&
{NC}		& {7200}& {NC}	& {4829.4}	&{NC} 	&{7200} & {0.156}	& {1.37}			\\ \hline
		
			\textbf{{Baron}}                     &{0.0002} & {00.74}	&{0.039}
		&{1002.14} & {0.078}	&{7200} 	& {0.117}	& {3.41}&{0.156} 	& 	{0.15}		\\ \hline
	\textbf{{Antigone}}                   & {0.0002}& {3.57}	&{0.039}
	&{0.38} & {0.081}	& {0.34}	& {0.117}	&{0.32} & {0.156}	&{0.13}			\\ \hline	
%		\textbf{Relaxed}                 & 0.0002	&0.201
%		&	0.039	&0.111	&0.078	&0.095&	0.117&	0.102&	0.156&	0.105
%		\\ \hline
	\end{tabular}
	\label{tab:results7node}
%	\vspace{-0.15in}
\end{table*}

\vspace{-0.15in}
\subsection{{Efficacy of the C$3$ Framework}}

%\vspace{-0.1in}
%\subsection{Energy Efficiency} \label{subsec:efficiency}
{Figure \ref{fig:mono2node} shows the increase in energy consumption with increase in the number of requests in different network and compression settings. We observe that as the number of requests increases, the total energy cost increases, as reflected in Remark~\ref{thm:increasingR}. An important observation is that the initial increase in the energy cost is large. However, when the data are cached (number of requests $\approx 40$), the slope decreases. This is because the transmission cost is usually much larger than the caching cost (using the energy proportional model for caching \cite{choi2012network}) and once the data are cached, the cached copy is used to satisfy other requests.} \par {For the energy efficiency, }
we compare the total energy costs under joint C$3$ optimization with those under C$2$ optimization. We consider two cases for the C$2$ optimization: (i) C$2$o (Communication and Computation), where we set $S_v=0$ for each node to avoid any data caching; (ii) C$2$a (Communication and Caching), where we set $\gamma=\sum_{k \in \mathcal{K}}y_k,$ which is equivalent to $\delta_v=1,$ $\forall v\in V$, i.e., no computation. Comparison between C$3$, C$2$o and C$2$a is shown in Figure~\ref{fig:c3vsc2}.  
{For the parameters that we used in simulation},  the energy cost for the C$3$ joint optimization is lower than that for  C$2$o optimization for the same parameter setting. {This highlights the improvement that can be achieved using C$3$ framework. }% This captures the tradeoff between caching, communication and computation. 
 In other words, although C$3$ incurs caching costs, it may significantly reduce the communication and computation, which in turn brings down total energy cost.   Similarly, C$3$ optimization outperforms C$2$a.  Using Equation~(\ref{eq:efficiency}), energy efficiency improves by as much as $88\%$ for the C$3$ framework when compared with the C$2$ formulation.  These trends are observed  in other candidate network topologies. 
Figure \ref{fig:c3vsc2-2node} shows the improvement  that C$3$ brings in comparison with C$2$ for a two nodes network. Using Equation~(\ref{eq:efficiency}), energy efficiency improves by as much as $70\%$ for the C$3$ framework when compared with the C$2$ formulation. The results for three nodes and four nodes networks are presented in Tables~\ref{tab:results3node} and~\ref{tab:results4node}. 

\vspace{-0.1in}
\begin{rem}
	Note that the above results are based on parameter values typically used in the literature, as shown in Table~\ref{tab:algcompareparameters}. From our analysis, it is clear that the larger the ratio between $\varepsilon_{vT}$ and $\varepsilon_{vR}$, $\varepsilon_{vC}$, the larger will be the improvement provided by our C$3$ formulation. %{}
\end{rem}
\vspace{-0.2in}

\subsection{{Comparison of Solvers}}
{We compare the performance of our proposed V-SBB with other MINLP solvers in terms of:}
\subsubsection{The Best Solution to the Objective Function}\label{subsec:value}
 We compare the performance of V-SBB with three other candidate solvers for  the networks in Figure~\ref{fig:networks}.   The results for two nodes and seven nodes networks are presented in Tables~\ref{tab:results2node} and~\ref{tab:results7node}.  We observe that {V-SBB, Bonmin, SCIP, Antigone, and Baron achieve comparable {objective function} value for larger values of $\gamma$, while V-SBB outperforms other algorithms for lower values of $\gamma$ (discussed in detail later). Furthermore, Bonmin and SCIP cannot generate a feasible solution even if it exists for some cases. Particularly, for Bonmin, }
 % V-SBB achieves the lowest {objective function} value comparable to Bonmin for larger values of $\gamma$, and significantly outperforms Bonmin for smaller values of $\gamma$, which we discuss in detail later.  
 %However, Bonmin cannot generate a feasible solution even if it exists for some cases. 
 {there are a number of probable reasons for such a problem: a) For MINLP problems with non-convex functions, Bonmin relies on heuristic options and does not guarantee $\epsilon$-global optimality \cite{fiat2009algorithms}. The heuristics can cause such problems; b) }  { The Branch-and-Cut method, used by Bonmin,  is based on outer-approximation (OA) algorithm \cite{bonminManual}. For the MINLP with non-convex functions, OA  constraints do not necessarily result in valid inequalities for the problem. Hence Bonmin's Branch-and-Cut method  sometimes cuts regions where a lower value exists.  }  NOMAD and GA in general yield a higher objective-function value than V-SBB does. This is because both NOMAD and GA are based on a stochastic approach which cannot guarantee convergence to the $\epsilon$-global optimum.  Similar trends are observed for three and four node networks. %, which are presented in Appendix~\ref{sec:appendixg}.

%Figure~\ref{fig:mono2node} verifies that the optimal energy cost is monotonically increasing with the number of requests, as stated in Remark~\ref{thm:increasingR} for a two node and seven node network. The results are obtained using our C3 framework for $\gamma=0.25\sum_{k \in \mathcal{K}}y_k$ and $\gamma=0.75\sum_{k \in \mathcal{K}}y_k,$ respectively. For the network parameters under consideration, we note that there is a turning point on the curves, and the total energy cost increases much faster with the number of requests before the turning point than that after it.   This is because the data has already been cached at the root node at this point and there is no need to retrieve data from other nodes in the network, which reduces transmission costs.  This is the benefit that caching brings {as mentioned earlier.} % and we will further discuss the advantage of C$3$ optimization over the C$2$ later in Section~\ref{subsec:efficiency}. 

%\vspace{-0.2in}
\subsubsection{Convergence Time}\label{subsec:speed}
The time taken to obtain the best solution is important in practice.  The amount of time that an algorithm requires to obtain its best solution as discussed in Section~\ref{subsec:value}  are shown in Tables~\ref{tab:results2node} and~\ref{tab:results7node} for the two nodes and seven nodes networks, respectively.  It  can be seen that Bonmin{, Antigone, Baron and SCIP (when it is able to provide a solution) are the fastest methods}.  {However,  Bonmin, SCIP and Baron sometimes  cannot find a solution although it exists.}%As discussed earlier, the Bonmin algorithm is fast at the expense of algorithm stability, i.e., sometimes it cannot find a solution although it exists.  
%This will be further discussed in the following section.  
V-SBB takes longer to obtain a better solution, because our reformulation introduces auxiliary variables and additional linear constraints.   Different applications can tolerate various degrees of algorithm speed. For the sample networks and applications under consideration, the speed of V-SBB is considered to be acceptable \cite{floudas2013deterministic}. 

%\vspace{-0.2in}
\subsubsection{Stability}\label{subsec:stability}
From the analysis in Sections~\ref{subsec:value} and~\ref{subsec:speed}, we know that Bonmin is faster but unstable in some situations.  We further characterize the stability of Bonmin with respect to the threshold value of QoI $\gamma$ as follows. Specifically, we fix all other parameters  in Table~\ref{tab:algcompareparameters}, and vary only the maximal possible value of $\gamma$ in different networks. The results are shown in Table~\ref{tab:bonmininfeas}. For each maximal value, we test all the possible integer values of $\gamma$ between $1$ and itself. Hence, the number of tests equals the maximal value.  We see that the number of instances where the Bonmin method fails to produce a feasible solution increases as the network size increases. % {This is due to the aforementioned reasons.} %This is mainly due to the cutting phase in the Bonmin method, which cuts the feasible regions that need to be branched. 

{Furthermore,  Bonmin, Baron and Antigone} can provide a feasible solution for smaller values of $\gamma$ at a faster time, we observe that the value of the solution is larger than that of V-SBB.  We compare the performance of V-SBB {with these algorithms} for smaller values of $\gamma$ in Table~\ref{tab:improvement}.  We see that V-SBB outperforms Bonmin{, Antigone and Baron} by as much as $52.45\%${, $50\%$, and $50\%$, respectively} when searching for an $\epsilon$-global optimum, though it requires more time. The timer is set to $7200$s for results shown in Table~\ref{tab:improvement}. {Results for three node and four node networks are given in Tables \ref{tab:results3node} and \ref{tab:results4node} respectively.  SCIP, for certain instances of the three node network, provides the lowest objective function value. However, for majority of the cases, we observe similar trends like Tables \ref{tab:results2node} and \ref{tab:results7node}. }

\begin{table}
%	\vspace{-0.1in}
	\centering
	\caption{Infeasibility of Bonmin for networks in Figure~\ref{fig:networks}}
%	\vspace{-0.1in}
	%\begin{tabular}{|l|l|l|l|l|}
	\begin{tabular}{|c|c|c|c|c|}
		\hline
		\textbf{Networks}  & (a) &  (b) &  (c) &  (d) \\ \hline
		\textbf{$\#$ of test values}      & {1000} & {2000} & {2000} & 4000 \\ \hline
		\textbf{$\#$ of infeasible solutions} & 0              & 0            & 1            & 216          \\ \hline
		\textbf{Infeasibility (\%)} & 0              & 0           & 0.05           & 5.4          \\ \hline
	\end{tabular}
%	\vspace{-0.3in}
	\label{tab:bonmininfeas}
\end{table}

%\red{With the increase in the number of nodes, Bonmin's performance deteriorates and it converges to an infeasible point. Even when Bonmin provides feasible solution for lower values of $\gamma$, the value of the solution is higher than the ones obtained with V-SBB. Table \ref{tab:improvement} compares the performance of V-SBB with Bonmin and shows that V-SBB outperforms Bonmin by as much as 52.45\% when searching for an $\epsilon-$global optimal; however it requires more time. }
\begin{table*}[ht]
	%\centering
	\footnotesize
	\caption{Comparison between V-SBB and Bonmin for smaller values of $\gamma$ in seven node network}
	\vspace{-0.1in}
	\label{tab:improvement}
	\begin{tabular}{|p{4cm}|l|l|l|l|l|l|l|l|l|l|}
		\hline
		\multirow{2}{*}{\textbf{Solver}} & \multicolumn{2}{c|}{$\gamma$ =1} & \multicolumn{2}{c|}{$\gamma$=3} & \multicolumn{2}{c|}{$\gamma$ =5} & \multicolumn{2}{c|}{$\gamma$ =8} & \multicolumn{2}{c|}{$\gamma$ =50} \\ \cline{2-11} 
		& \textbf{Obj.}   & \textbf{Time (s)}  & \textbf{Obj.}    & \textbf{Time}    & \textbf{Obj.}     & \textbf{Time}    & \textbf{Obj.}     & \textbf{Time}    & \textbf{Obj.}     & \textbf{Time}     \\ \hline
		\textbf{Bonmin}                  & 0.0002 &	0.214 &	0.0003 &	0.211 &	0.0003 &	0.224 &	 0.0005&	0.23&	0.0021&	0.364
		\\ \hline
			\textbf{{Antigone}}                  & {0.0002} &	{3.57} &	{0.000317} &	{2.47} &	{0.000395} &	{6.53} &	{0.000512} & {15.61} & {0.002153} & {2.71}
		\\ \hline
			\textbf{{Baron}}                  & {0.0002} &	{0.74} &	{0.00031} &	{4846} &	{0.00039} &	{7200} &	{0.005} & {7200} & {0.0021} & {7200}
		\\ \hline
		
		\textbf{V-SBB}                   & 0.00011&	1871&	0.00015&	2330&	0.00019&	1243	&0.00047	&1350&	0.0020	&3325
		\\ \hline
		\textbf{Improvement over Bonmin (\%)}             & \multicolumn{2}{c|}{52.45}           & \multicolumn{2}{c|}{49.43}          & \multicolumn{2}{c|}{50.30}          & \multicolumn{2}{c|}{7.59}         & \multicolumn{2}{c|}{4.62}   \\ \hline
	
		\textbf{{Improvement over Antigone (\%)}}             & \multicolumn{2}{c|}{{50
		}}           & \multicolumn{2}{c|}{{52.72}}          & \multicolumn{2}{c|}{{51.92}}          & \multicolumn{2}{c|}{{8.27}}         & \multicolumn{2}{c|}{{7.08}}   \\ \hline
		
		\textbf{{Improvement over Baron (\%)}}             & \multicolumn{2}{c|}{{50}}           & \multicolumn{2}{c|}{{51.61}}          & \multicolumn{2}{c|}{{51.28}}          & \multicolumn{2}{c|}{{6}}         & \multicolumn{2}{c|}{{4.79}}   \\ \hline

	\end{tabular}
%	\vspace{-0.1in}
\end{table*}

\begin{figure}
	\centering
%	\vspace{-0.2in}
	\includegraphics[width=0.8\textwidth]{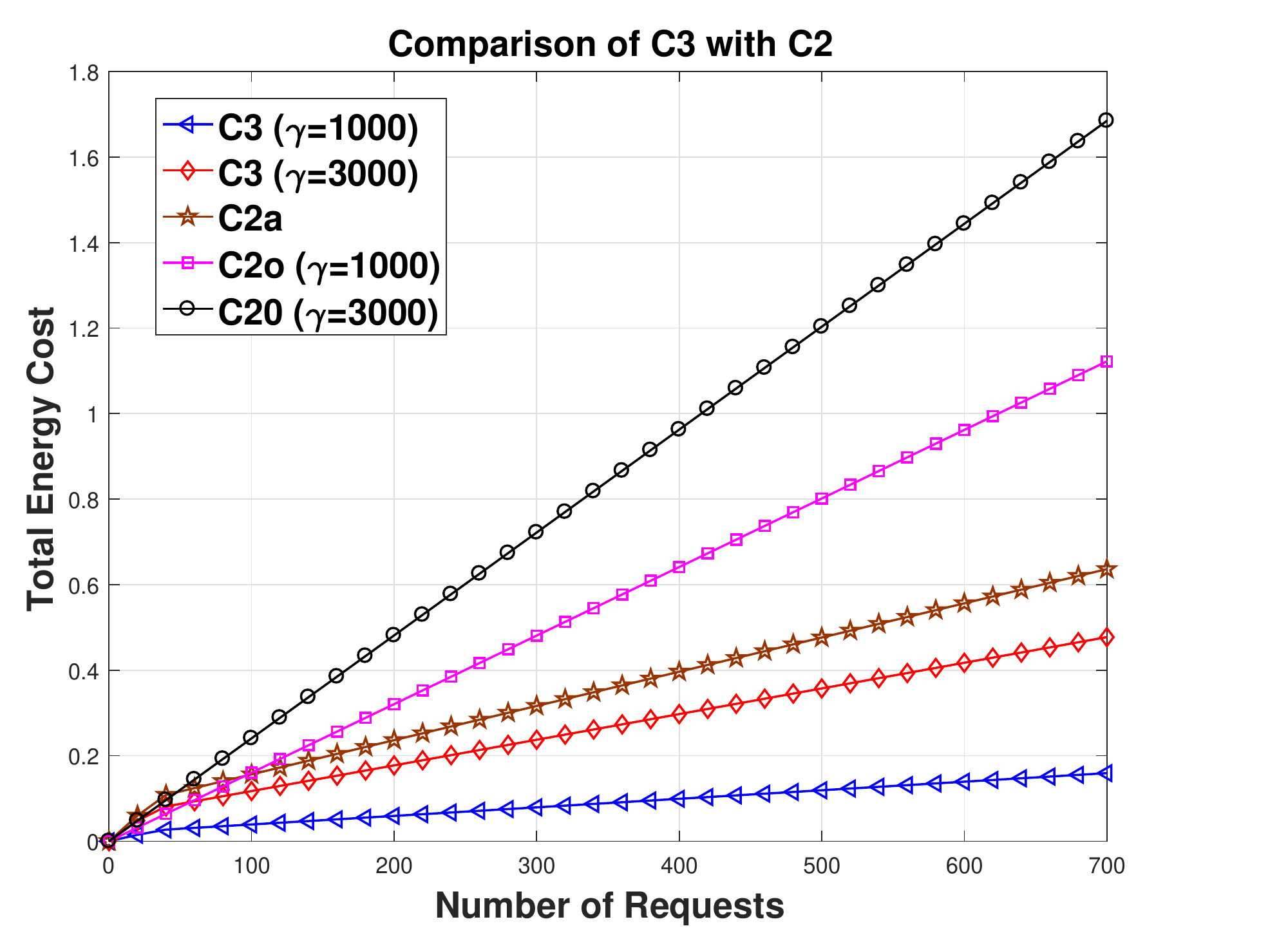}
%	\vspace{-0.3in}
	\caption{Comparison of C$3$ and C$2$ optimization for the seven node network in Figure~\ref{fig:networks}.}
	\label{fig:c3vsc2}
	\vspace{-0.2in}
\end{figure} 

\vspace{-0.1in}
\begin{figure}
	\centering
	\includegraphics[width=0.8\textwidth]{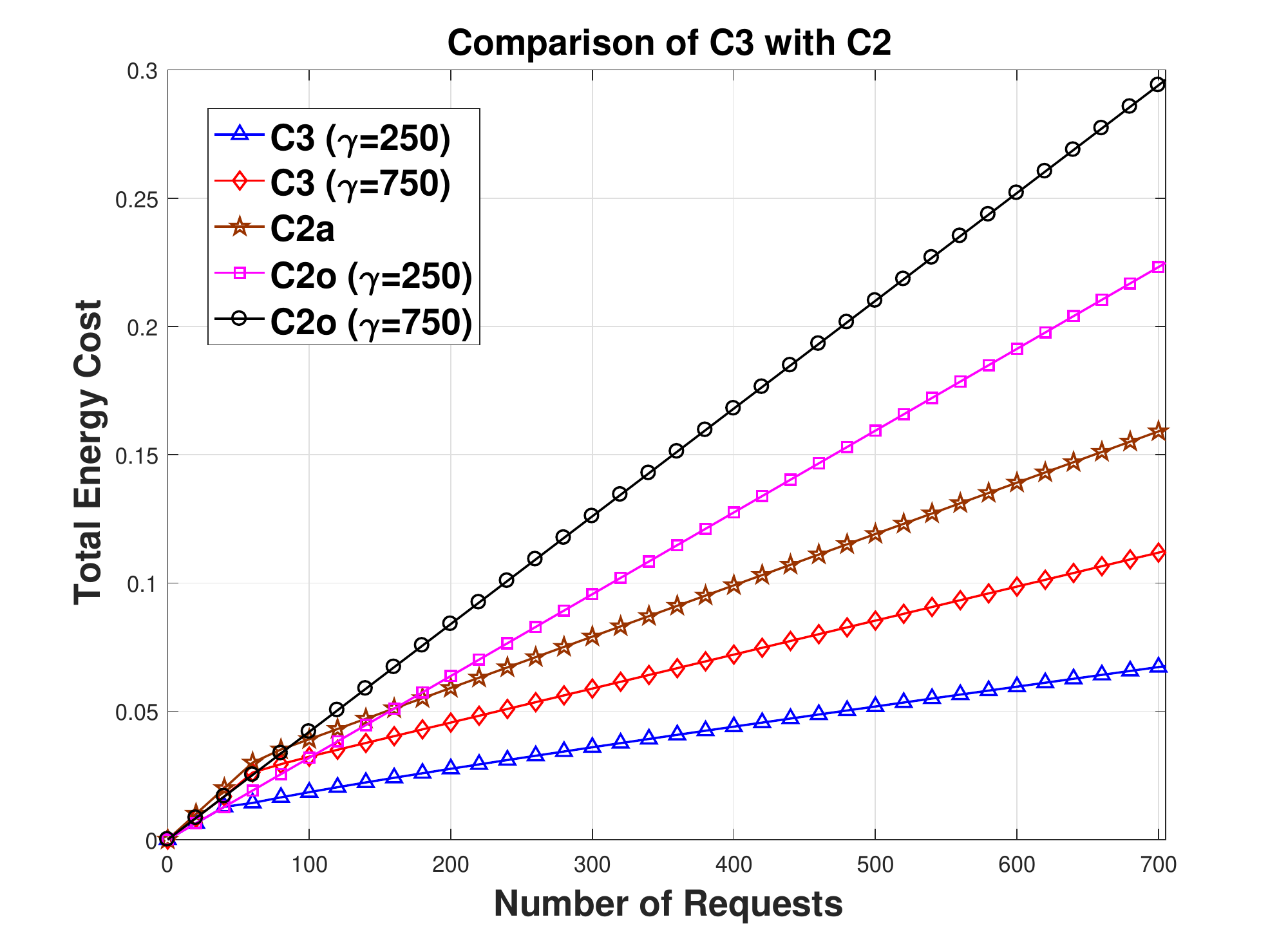}
%	\vspace{-0.3in}
	\caption{Comparison of C$3$ and C$2$ optimization for the two nodes network in Figure~\ref{fig:networks}.}
	\label{fig:c3vsc2-2node}
	\vspace{-0.3in}
\end{figure} 

% \usepackage{multirow}

%\subsection{Advantage of using joint communication, computation and caching} 

\begin{table*}[]
	\centering
	\caption{The Value of Objective Function (Obj.) and Convergence Speed for three node network}
	\vspace{-0.1in}
	\begin{tabular}{|l|l|l|l|l|l|l|l|l|l|l|}
		\hline
		\multirow{2}{*}{\textbf{Solver}} & \multicolumn{2}{c|}{{$\gamma=1$}}               & \multicolumn{2}{c|}{{{$\gamma= 500$}}}      & \multicolumn{2}{c|}{{{$\gamma=1000$}}}     & \multicolumn{2}{c|}{\textbf{{$\gamma=1500$}}}     & \multicolumn{2}{c|}{\textbf{{$\gamma=2000$}}}    \\ \cline{2-11} 
		& \textbf{Obj.} & \textbf{Time} (s) 	& \textbf{Obj.} & \textbf{Time} (s) & \textbf{Obj.} & \textbf{Time} (s)& \textbf{Obj.} & \textbf{Time} (s) & \textbf{Obj.} & \textbf{Time} (s) \\ \hline
		\textbf{Bonmin}                 & {0.005}&{0.26}	&{0.01}
		&{0.14}&{0.019}	&{0.10}	&{0.028}	&{0.10}	&{0.0383}	&{5.56}
		
		\\ \hline
		\textbf{NOMAD}                  & {0.045}&{12.42}	&{0.025}
		&{11.14}&{0.033}	&{9.30}	&{0.029}	&{46.41}	&{0.038}	&{5.45}
		
		\\ \hline
		\textbf{GA}                     & {0.005}&{0.69}	&{0.025}
		&{26.11}&{0.019}	&{16.85}	&{0.034}	&{40.56}	&{0.044}	&{10.34}
		
		\\ \hline
		\textbf{V-SBB}           & {0.005}&{46.1}	&{0.019}	&{45.34}&{0.019}	&{8.1}	&{0.028}	&{56.3}	&{0.0383}	&{12.2}
		
		\\ \hline
		
		\textbf{{SCIP}}                   & {0.00005}&{4.96}	&{0.000056}
		&{0.16}&{0.000054}	&{0.18}	&{0.028}	&{0.07}	&{0.038}	&{0.05}	\\ \hline
		
			\textbf{{Baron}}                     &{0.005} & {0.1}	&{0.01}
		&{0.09} & {0.019}	&{0.09} 	& {0.028}	& {0.1}&{0.0383} 	& 	{0.1}		\\ \hline
		\textbf{{Antigone}}                   & {0.005 }& {0.11}	&{0.01}
		&{0.09} & {0.019}	& {0.08}	& {0.028}	&{0.21} & {0.038}	&{1.51}			\\ \hline
%		\textbf{Relaxed}            & {0.005}&{0.068}	&{0.01}
	%	&{0.05}&{0.019}	&{0.060}	&{0.028}	&{0.05}	&{0.038}	&{0.04}
		
%		\\ \hline
	\end{tabular}
	%	\vspace{-0.2in}
	\label{tab:results3node}
\end{table*}

\begin{table*}[]
	\centering
	\caption{The Value of Objective Function (Obj.) and Convergence Speed for four node network}
	\vspace{-0.1in}
	\begin{tabular}{|l|l|l|l|l|l|l|l|l|l|l|}
		\hline
		\multirow{2}{*}{\textbf{Solver}} & \multicolumn{2}{c|}{{$\gamma=1$}}               & \multicolumn{2}{c|}{{{$\gamma= 500$}}}      & \multicolumn{2}{c|}{{{$\gamma=1000$}}}     & \multicolumn{2}{c|}{\textbf{{$\gamma=1500$}}}     & \multicolumn{2}{c|}{\textbf{{$\gamma=2000$}}}    \\ \cline{2-11} 
		& \textbf{Obj.} & \textbf{Time} (s) & \textbf{Obj.} &  \textbf{Time} (s) & \textbf{Obj.} & \textbf{Time} (s)& \textbf{Obj.} & \textbf{Time} (s) & \textbf{Obj.} & \textbf{Time} (s) \\ \hline
		\textbf{Bonmin}                  & {0.002}&{0.36}	&{0.02}
		&{0.11}&{0.039}	&{ 0.11}	&{0.06}	&{ 0.10}	&{0.08}	&{0.16}
		\\ \hline
		\textbf{NOMAD}                      & {0.003}&{112.5}	&{0.023}
		&{97.68}&{0.04}	&{59.86}	&{0.06}	&{52.8}	&{0.10}	&{ 2.28}
		
		\\ \hline
		\textbf{GA}                       & {0.004}&{1.01}	&{0.02}
		&{24.94}&{0.04}	&{13.02}	&{ 0.12}	&{27.7}	&{0.14}	&{ 35.33}
		
		\\ \hline
		\textbf{V-SBB}                     &{0.02} &{400}	&
{0.02}	&{400}		& {0.039}&{400}	&{0.071}	& {400}		& {0.078}&	{400}
		
		\\ \hline
		
		\textbf{{SCIP}}                & {0.002}&{7.05}	&{0.02}
		&{1999.4}&{0.0004}	&{2.00}	&{0.009}	&{0.43}	&{0.04}	&{0.16}		\\ \hline
		
			\textbf{{Baron}}                     &{0.002} & {0.52}	&{0.02}
		&{2.69} & {0.039}	&{0.89} 	& {0.06}	& {0.16}&{0.078} 	& 	{0.1}		\\ \hline
	\textbf{{Antigone}}                   & {0.002 }& {21.2}	&{0.02}
	&{0.26} & {0.042}	& {0.18}	& {0.06}	&{0.1} & {0.078}	&{0.08}			\\ \hline	
%		\textbf{Relaxed}                     & {0.002}&{0.57}	&{ 0.02}
%		&{ 0.06}&{0.039}	&{0.065}	&{0.0}	&{0.06}	&{0.08}	&{0.11}
	%	\\ \hline
	\end{tabular}
	\label{tab:results4node}
		\vspace{-0.25in}
\end{table*}

\section{Conclusion}\label{sec:con}

We have investigated energy efficiency tradeoffs among communication, computation and caching with QoI guarantee in {distributed} networks. We first formulated an optimization problem that characterizes these energy costs.  This optimization problem belongs to the non-convex class of MINLP, which is hard to solve in general.  We then proposed a variant of the  spatial branch-and-bound (V-SBB) algorithm, which can solve the MINLP with $\epsilon$-optimality guarantee.  Finally, we showed numerically that the newly proposed V-SBB algorithm outperforms the existing MINLP solvers, Bonmin, NOMAD and GA.  We also observed that C3 optimization framework, which to the best of our knowledge has not been investigated in the literature,  leads to an energy saving of as much as $88\%$ {when} compared with either of the C2 optimizations which have been widely studied. 

Going further, we aim to extend our results in two ways.  The first is to refine and improve the symbolic reformulation to reduce the number of needed auxiliary variables in order to shorten the algorithm execution time. Second, since many networking problems involve the optimization of both continuous and discrete variables as in this work, we plan to apply and extend the newly proposed V-SBB to solve those problems.

\section*{Acknowledgments}\label{sec:ack}
%This work was supported by the U.S. Army Research Laboratory and the U.K. Ministry of Defence under Agreement Number W911NF-16-3-0001. The views and conclusions contained in this document are those of the authors and should not be interpreted as representing the official policies, either expressed or implied, of the U.S. Army Research Laboratory, the U.S. Government, the U.K. Ministry of Defence or the U.K. Government. The U.S. and U.K. Governments are authorized to reproduce and distribute reprints for Government purposes notwithstanding any copy-right notation hereon. Faheem Zafari also acknowledges the financial support by EPSRC Centre for Doctoral Training in High Performance Embedded and Distributed Systems  (HiPEDS, Grant Reference EP/L016796/1), and Department of Electrical and Electronics Engineering, Imperial College London.
The material in this paper has been accepted for publication in part at IEEE Globecom, Abu Dhabi, United Arab Emirates, December 2018.
This work was supported by the U.S. Army Research Laboratory and the U.K. Ministry of Defence under Agreement Number W911NF-16-3-0001.  The views and conclusions contained in this document are those of the authors and should not be interpreted as representing the official policies, either expressed or implied, of the U.S. Army Research Laboratory, the U.S. Government, the U.K. Ministry of Defence or the U.K. Government. The U.S. and U.K. Governments are authorized to reproduce and distribute reprints for Government purposes notwithstanding any copy-right notation hereon. Faheem Zafari also acknowledges the financial support by EPSRC Centre for Doctoral Training in High Performance Embedded and Distributed Systems  (HiPEDS, Grant Reference EP/L016796/1), and Department of Electrical and Electronics Engineering, Imperial College London. {The authors will also like to thank Dr. Ruth Misener and the Chemical Engineering Department at Imperial College London for providing us with the access to Baron and Antigone Solvers.

%\vspace{-8pt}
\bibliographystyle{IEEETran}
\bibliography{refs2}  

\begin{thebibliography}{10}
\providecommand{\url}[1]{#1}
\csname url@samestyle\endcsname
\providecommand{\newblock}{\relax}
\providecommand{\bibinfo}[2]{#2}
\providecommand{\BIBentrySTDinterwordspacing}{\spaceskip=0pt\relax}
\providecommand{\BIBentryALTinterwordstretchfactor}{4}
\providecommand{\BIBentryALTinterwordspacing}{\spaceskip=\fontdimen2\font plus
\BIBentryALTinterwordstretchfactor\fontdimen3\font minus
  \fontdimen4\font\relax}
\providecommand{\BIBforeignlanguage}[2]{{%
\expandafter\ifx\csname l@#1\endcsname\relax
\typeout{** WARNING: IEEEtran.bst: No hyphenation pattern has been}%
\typeout{** loaded for the language `#1'. Using the pattern for}%
\typeout{** the default language instead.}%
\else
\language=\csname l@#1\endcsname
\fi
#2}}
\providecommand{\BIBdecl}{\relax}
\BIBdecl

\bibitem{nazemi2016qoi}
S.~Nazemi, K.~K. Leung, and A.~Swami, ``Qo{I}-aware {T}radeoff {B}etween
  {C}ommunication and {C}omputation in {W}ireless {A}d-hoc {N}etworks,'' in
  \emph{Proc. IEEE PIMRC}, 2016.

\bibitem{rajagopalan06}
R.~Rajagopalan and P.~K. Varshney, ``Data {A}ggregation {T}echniques in
  {S}ensor {N}etworks: A {S}urvey,'' \emph{IEEE Commun. Surveys Tuts.}, vol.~8,
  no.~4, pp. 48–--63, 2006.

\bibitem{fasolo07}
E.~Fasolo, M.~Rossi, J.~Widmer, and M.~Zorzi, ``In-network {A}ggregation
  {T}echniques for {W}ireless {S}ensor {N}etworks: a {S}urvey,'' \emph{IEEE
  Wireless Communications}, vol.~14, no.~2, 2007.

\bibitem{barr2006energy}
K.~C. Barr and K.~Asanovi{\'c}, ``Energy-aware {L}ossless {D}ata
  {C}ompression,'' \emph{ACM Transactions on Computer Systems}, 2006.

\bibitem{choi2012network}
N.~Choi, K.~Guan, D.~C. Kilper, and G.~Atkinson, ``In-network {C}aching
  {E}ffect on {O}ptimal {E}nergy {C}onsumption in {C}ontent-{C}entric
  {N}etworking,'' in \emph{Proc. IEEE ICC}, 2012.

\bibitem{ehikioya1999characterization}
S.~A. Ehikioya, ``A {C}haracterization of {I}nformation {Q}uality {U}sing
  {F}uzzy {L}ogic,'' in \emph{NAFIPS}, 1999.

\bibitem{smith1996global}
E.~M. Smith and C.~C. Pantelides, ``Global {O}ptimisation of {G}eneral
  {P}rocess {M}odels,'' in \emph{Glo. Opt. Eng. Des.}\hskip 1em plus 0.5em
  minus 0.4em\relax Springer, 1996, pp. 355--386.

\bibitem{bonami2008algorithmic}
P.~Bonami \emph{et~al.}, ``An {A}lgorithmic {F}ramework for {C}onvex {M}ixed
  {I}nteger {N}onlinear {P}rograms,'' \emph{Disc. Opt.}, vol.~5, no.~2, pp.
  186--204, 2008.

\bibitem{le2011algorithm}
S.~Le~Digabel, ``Algorithm 909: {NOMAD}: {N}onlinear {O}ptimization with the
  {MADS} {A}lgorithm,'' \emph{ACM TOMS}, vol.~37, no.~4, p.~44, 2011.

\bibitem{tawarmalani2005polyhedral}
M.~Tawarmalani and N.~V. Sahinidis, ``A polyhedral branch-and-cut approach to
  global optimization,'' \emph{Mathematical Programming}, vol. 103, no.~2, pp.
  225--249, 2005.

\bibitem{achterberg2009scip}
T.~Achterberg, ``{SCIP}: {S}olving {C}onstraint {I}nteger {P}rograms,''
  \emph{Mathematical Programming Computation}, vol.~1, no.~1, pp. 1--41, 2009.

\bibitem{misener2014antigone}
R.~Misener and C.~A. Floudas, ``Antigone: algorithms for continuous/integer
  global optimization of nonlinear equations,'' \emph{Journal of Global
  Optimization}, vol.~59, no. 2-3, pp. 503--526, 2014.

\bibitem{boykin14}
O.~Boykin, S.~Ritchie, I.~O'Connell, and J.~Lin, ``Summingbird: {A} {F}ramework
  for {I}ntegrating {B}atch and {O}nline {M}apreduce {C}omputations,''
  \emph{Proc. of VLDB}, 2014.

\bibitem{jian17}
J.~Li, S.~Shakkottai, J.~C.~S.~Lui, and V.~Subramanian, ``{Accurate {L}earning
  or {F}ast {M}ixing? {D}ynamic {A}daptability of {C}aching {A}lgorithms},''
  \emph{IEEE Journal on Selected Areas in Communications}, 2018.

\bibitem{ioannidis16}
S.~Ioannidis and E.~Yeh, ``Adaptive {C}aching {N}etworks with {O}ptimality
  {G}uarantees,'' in \emph{Proc. of ACM SIGMETRICS}, 2016.

\bibitem{jianfaheem18icpe}
J.~Li, F.~Zafari, D.~Towsley, K.~K.~Leung, and A.~Swami, ``Joint {D}ata
  {C}ompression and {C}aching: {A}pproaching {O}ptimality with {G}uarantees,''
  in \emph{Proc. of ACM/SPEC ICPE}, 2018.

\bibitem{heinzelman2000energy}
W.~R. Heinzelman, A.~Chandrakasan, and H.~Balakrishnan, ``Energy-{E}fficient
  {C}ommunication {P}rotocol for {W}ireless {M}icrosensor {N}etworks,'' in
  \emph{System sciences}, 2000.

\bibitem{manjeshwar2001teen}
A.~Manjeshwar and D.~P. Agrawal, ``{TEEN}: a {R}outing {P}rotocol for
  {E}nhanced {E}fficiency in {W}ireless {S}ensor {N}etworks,'' in \emph{IPDPS},
  2001.

\bibitem{ye2005eecs}
M.~Ye, C.~Li, G.~Chen, and J.~Wu, ``{EECS}: an {E}nergy {E}fficient
  {C}lustering {S}cheme in {W}ireless {S}ensor {N}etworks,'' in \emph{Proc. of
  IEEE IPCCC}, 2005.

\bibitem{laporta2012}
S.~Eswaran, J.~Edwards, A.~Misra, and T.~F.~L. Porta, ``Adaptive {I}n-{N}etwork
  {P}rocessing for {B}andwidth and {E}nergy {C}onstrained {M}ission-{O}riented
  {M}ultihop {W}ireless {N}etworks,'' \emph{IEEE Transactions on Mobile
  Computing}, vol.~11, no.~9, pp. 1484--1498, Sept 2012.

\bibitem{smith1999symbolic}
E.~M. Smith and C.~C. Pantelides, ``A {S}ymbolic {R}eformulation/{S}patial
  {B}ranch-and-{B}ound {A}lgorithm for the {G}lobal {O}ptimisation of
  {N}onconvex {MINLP}s,'' \emph{Comp. \& Chem. Eng.}, vol.~23, no.~4, pp.
  457--478, 1999.

\bibitem{smith1996optimal}
E.~M. Smith, ``On the {O}ptimal {D}esign of {C}ontinuous {P}rocesses,'' Ph.D.
  dissertation, Imperial College London (University of London), 1996.

\bibitem{liberti2004reformulation}
L.~Liberti, ``Reformulation and {C}onvex {R}elaxation {T}echniques for {G}lobal
  {O}ptimization,'' \emph{4OR: A Quarterly Journal of Operations Research},
  vol.~2, no.~3, pp. 255--258, 2004.

\bibitem{liberti2004phdthesis}
L.~Liberti, ``Reformulation and {C}onvex {R}elaxation {T}echniques for {G}lobal
  {O}ptimization,'' Ph.D. dissertation, Imperial College London, 2004.

\bibitem{floudas2013deterministic}
C.~A. Floudas, \emph{Deterministic {G}lobal {O}ptimization: {T}heory, {M}ethods
  and {A}pplications}.\hskip 1em plus 0.5em minus 0.4em\relax Springer Science
  \& Business Media, 2013, vol.~37.

\bibitem{mccormick1976computability}
G.~P. McCormick, ``Computability of {G}lobal {S}olutions to {F}actorable
  {N}onconvex {P}rograms: {P}art {I}—{C}onvex {U}nderestimating {P}roblems,''
  \emph{Mathematical Programming}, vol.~10, no.~1, pp. 147--175, 1976.

\bibitem{deb2002fast}
K.~Deb, A.~Pratap, S.~Agarwal, and T.~Meyarivan, ``A {F}ast and {E}litist
  {M}ultiobjective {G}enetic {A}lgorithm: {NSGA-II},'' \emph{IEEE transactions
  on evolutionary computation}, vol.~6, no.~2, pp. 182--197, 2002.

\bibitem{optitoolbox}
{OPTI Toolbox}, ``{{A} {F}ree {M}atlab {T}oolbox for {O}ptimization},''
  \url{https://www.inverseproblem.co.nz/OPTI/index.php/Main/HomePage}, [Online;
  accessed 28-Jun-2017].

\bibitem{wachter06}
A.~W{\"a}chter and L.~T. Biegler, ``On the {I}mplementation of an
  {I}nterior-point {F}ilter {L}ine-search {A}lgorithm for {L}arge-scale
  {N}onlinear {P}rogramming,'' \emph{Mathematical Programming}, vol. 106,
  no.~1, pp. 25--57, 2006.

\bibitem{ye2002energy}
W.~Ye, J.~Heidemann, and D.~Estrin, ``An {E}nergy-{E}fficient {MAC} {P}rotocol
  for {W}ireless {S}ensor {N}etworks,'' in \emph{Proc. of IEEE INFOCOM}, 2002.

\bibitem{fiat2009algorithms}
A.~Fiat and P.~Sanders, ``Algorithms-esa 2009,'' \emph{Lecture Notes in
  Computer Science}, vol. 5757, 2009.

\bibitem{bonminManual}
P.~Bonami and J.~Lee, ``{BONMIN Users' Manual},''
  \url{https://projects.coin-or.org/Bonmin/browser/stable/1.5/Bonmin/doc/BONMIN_UsersManual.pdf?format=raw},
  2011.

\end{thebibliography}

%\clearpage
\appendices
\section{}\label{sec:appendixd}

\begin{small}
	\begin{align}
	\min_{w}\quad &w_{f}\nonumber \\
	\text{s.t.} \quad&\sum_{k\in \mathcal{K}}y_{k} {\overline{w}_{k,j}^{C_1}}\geq\gamma, \nonumber\\
	&  \sum_{k \in C_v}y_k {\overline{w}_{k,i}^{C_2}}\leq S_v, \forall \; v \in V, \nonumber \\
	& \sum_{i=0}^{h(k)}b_{k,i} \leq 1,\forall k\in\mathcal{K},  \nonumber \\
	&  b_{k,i}\in\{0,1\}, \forall  k\in\mathcal{K}, i=0, \cdots, h(k),  \nonumber\\
	&	\mathbf{w}^l\leq \mathbf{w}\leq \mathbf{w}^U,\; \forall  k\in\mathcal{K}, i=0, \cdots, h(k),  \nonumber \\
	& w_{k,i}^b=\delta_{k, i}\times \overline{w}_{k,a}, \;  \forall  k\in\mathcal{K}, i=0, \cdots, h(k),  \nonumber \\
	& w_{k,i}^f=\frac{\overline{w}_{k,a}}{\delta_{k, i}} , \;  \forall  k\in\mathcal{K}, i=0, \cdots, h(k), \nonumber \\
	& y_k b_{k,i}-\overline{w}_{k,i}^{C''_2}=0 , \;  \forall  k\in\mathcal{K}, i=0, \cdots, h(k), \nonumber\\
	& \overline{w}_{k,i}^{C_2}= \overline{w}_{k,i}^{C''_2} \times \overline{w}_{k,\beta}^{C'_2} , \;  \forall  k\in\mathcal{K}, i=0, \cdots, h(k), \nonumber    \\
	& \sum_{j=0}^{i-1}b_{k,j} - \widetilde{w}_{k,i}=0 , \;  \forall  k\in\mathcal{K}, i=0, \cdots, h(k), \nonumber \\
	& \overline{\overline{w}}_{k,i}=\overline{w}_{k,a} \times \widetilde{w}_{k,i} , \;  \forall  k\in\mathcal{K}, i=0, \cdots, h(k),  \nonumber   \\
	& \overline{w}_{k,i}^b=w_{k,i}^b \times \widetilde{w}_{k,i} , \;  \forall  k\in\mathcal{K}, i=0, \cdots, h(k),  \nonumber \\
	& \overline{w}_{k,i}^f=w_{k,i}^f \times \widetilde{w}_{k,i}, \;  \forall  k\in\mathcal{K}, i=0, \cdots, h(k),  \nonumber \\
	& \prod_{m=i+1}^{h(k)}\delta_{k, m}=\overline{w}_{k,\underbrace{h(k)-2-i}_a}=\begin{cases}
	\delta_{k,h(k)} \times \delta_{k,h(k)-1},& \forall a=0 \\ 
	\overline{w}_{k,a-1}\times \delta_{k,m}, &  m+a=h(k)-1 \\ 
	\delta_{k,h(k)}, & \forall i=h(k)-1, 
	\end{cases} \nonumber \\
	& \prod_{i=0}^{h(k)}\delta_{k, i}=\overline{w}_{k,\underbrace{h(k)-1}_j}^{C_1}= 
	\begin{cases}
	\delta_{k,h(k)} \times \delta_{k,h(k)-1},& \forall j=0 \\ 
	\overline{w}_{k,j-1}^{C_1}\times \delta_{k,j+1}, &  \forall j=1 \cdots h(k)-1 \\ 
	\delta_{k,h(k)}, & \forall i=h(k),
	\end{cases} \nonumber\\
	& \prod_{j=h(k)}^{h(v)}\delta_{k,j}={\overline{w}_{k,\underbrace{\tau-2}_\beta}^{C'_2}=}
	\begin{cases}
	\delta_{k,h(k)} \times \delta_{k,h(k)-1},& \forall \beta=0 \\ 
	\overline{w}_{k,j-1}^{C_1}\times \delta_{k,j+1}, &  \forall \beta > 0 \\ 
	\delta_{k,h(k)}, & \forall \beta <0,
	\end{cases} \nonumber \\
	&w_f=\sum_{k \in \mathcal{K}}\sum_{i=0}^{h(k)}y_k \Bigg(\bigg( \varepsilon_{kR}\overline{w}_{k,a}+\varepsilon_{kT}w_{k,i}^b+\varepsilon_{kC} w_{k,i}^f - \varepsilon_{kc}\overline{w}_{ka}\bigg)+A+B\Bigg),\nonumber \\
	& A=\varepsilon_{kR}R_k\overline{w}_{k,a}+\varepsilon_{kT}w_{k,i}^b+\varepsilon_{kc}R_kw_{k,i}^f-\varepsilon_{kC}R_k\overline{w}_{k,a}- \varepsilon_{kR}\overline{w}_{k,a}-\varepsilon_{kT}w_{k,i}^b-\varepsilon_{kC}w_{k,i}^f+\varepsilon_{kC}\overline{w}_{k,a}, \nonumber \\
	& B=-\varepsilon_{kR}R_k \overline{\overline{w}}_{k,i}-\varepsilon_{kT}R_k \overline{w}_{k,i}^b-\varepsilon_{kc}R_k\overline{w}_{k,i}^f+  \varepsilon_{kC}R_k\overline{\overline{w}}_{k,i}+\varepsilon_{kR}\overline{\overline{w}}_{k,i}+\varepsilon_{kT}\overline{w}_{k,i}^b+\varepsilon_{kC}\overline{w}_{k,i}^f- \varepsilon_{kC}\overline{\overline{w}}_{k,i}.
	\label{eq:exactproblem}
	\end{align}
\end{small}

%\newpage

\section{}\label{sec:appendixc}

\begin{table*}
%	\small
	%\centering
	\caption{Symbolic Reformulation Rules defined in \cite{smith1996global} where X stands for expression, C stands for Constant and V stands for variable}
	\begin{tabular}{|c|c|p{3cm}|l|p{2cm}|p{2cm}|}
		\hline
		\multicolumn{1}{|l|}{\textbf{Left Subtree Class}} & \multicolumn{1}{l|}{\textbf{Right Subtree Class}} & \textbf{Binary Operator} & \textbf{New Variable Definition} & \textbf{New Linear Constraint } & \textbf{Binary Tree Class} \\ \hline
		\multirow{3}{*}{C}         & \multirow{3}{*}{C}       & $\pm$      &                  &     & C    \\ \cline{3-6} 
		&                                                   & $\times$                       &                                  &                                         & C                          \\ \cline{3-6} 
		&                                                   & $\div$                        &                                  &                                         & C                          \\ \hline
		\multirow{3}{*}{V}                                & \multirow{3}{*}{C}                                & $\pm$                      &                                  &                                         & X                          \\ \cline{3-6} 
		&                                                   & $\times$                        &                                  &                                         & X                          \\ \cline{3-6} 
		&                                                   & $\div$                        &                                  &                                         & X                          \\ \hline
		\multirow{3}{*}{X}                                & \multirow{3}{*}{C}                                & $\pm$                      &                                  &                                         & X                          \\ \cline{3-6} 
		&                                                   & $\times$                        &                                  &                                         & X                          \\ \cline{3-6} 
		&                                                   & $\div$                        &                                  &                                         & X                          \\ \hline
		\multirow{3}{*}{C}                                & \multirow{3}{*}{V}                                & $\pm$                      &                                  &                                         & X                          \\ \cline{3-6} 
		&                                                   & $\times$                        &                                  &                                         & X                          \\ \cline{3-6} 
		&                                                   & $\div$                        & Linear Fractional                &                                         & V                          \\ \hline
		\multirow{3}{*}{V}                                & \multirow{3}{*}{V}                                & $\pm$                      &                                  &                                         & X                          \\ \cline{3-6} 
		&                                                   & $\times$                        & Bilinear                         &                                         & V                          \\ \cline{3-6} 
		&                                                   & $\div$                        & Linear Fractional                &                                         & V                          \\ \hline
		\multirow{3}{*}{X}                                & \multirow{3}{*}{V}                                & $\pm$                      &                                  &                                         & X                          \\ \cline{3-6} 
		&                                                   & $\times$                        & Bilinear                         & Left                                    & X                          \\ \cline{3-6} 
		&                                                   & $\div$                        & Linear Fractional                & Left                                    & X                          \\ \hline
		\multirow{3}{*}{C}                                & \multirow{3}{*}{X}                                & $\pm$                      &                                  &                                         & X                          \\ \cline{3-6} 
		&                                                   & $\times$                         &                                  &                                         & X                          \\ \cline{3-6} 
		&                                                   & $\div$                        & Linear Fractional                & \multicolumn{1}{r|}{Right}              & V                          \\ \hline
		\multirow{3}{*}{V}                                & \multirow{3}{*}{X}                                & $\pm$                      &                                  &                                         & X                          \\ \cline{3-6} 
		&                                                   & $\times$                        & Bilinear                         & \multicolumn{1}{r|}{Right}              & V                          \\ \cline{3-6} 
		&                                                   & $\div$                        & Linear Fractional                & \multicolumn{1}{r|}{Right}              & V                          \\ \hline
		\multirow{3}{*}{X}                                & \multirow{3}{*}{X}                                & $\pm$                      &                                  &                                         & X                          \\ \cline{3-6} 
		&                                                   & $\times$                       & Bilinear                         & Left,  Right                            & V                          \\ \cline{3-6} 
		&                                                   & $\div$                        & Linear Fractional                & Left, Right                             & V                          \\ \hline
	\end{tabular}
	\label{tab:binaryrules}
\end{table*}

\subsection{Symbolic Reformulation}
%\vspace{-30pt}
  The first step of the symbolic reformulation is to represent the algebraic expression (objective function and constraints) using a binary tree as shown in Figure \ref{fig:symbinary}. Symbolic reformulation transforms the algebraic expression represented as binary tree into a set of linear constraints that might involve some newly introduced variables. As our optimization problem \eqref{eq:optimization}  contains bilinear and linear fractional terms, the newly introduced auxiliary variables are therefore either products or ratios of other variables i.e. $w_i\equiv w_jw_k$ and $w_i\equiv \frac{w_j}{w_k}$. The rules for efficiently\footnote{Keeping the number of newly introduced variables to minimum} achieving such transformation are presented in \cite{smith1996global} part of which we restate in the Table \ref{tab:binaryrules}. % given in Appendix \ref{sec:appendixc}.
   We create binary tree for representing the algebraic expressions and assign the leaf nodes a class that can be either a constant (C), an expression (X), or a variable (V).  If we are at some intermediate node that represents a multiplication operation, and both its right and left child nodes are of class expression (X), then the reformulation would require us to introduce two linear constraints (for both right and left node explained) as well as introduce new bilinear auxiliary variable. 
\begin{align}
\min_{\boldsymbol w}\quad &\boldsymbol w_{f}\nonumber \\
\text{s.t.} \quad&\mathbf{Aw=b}, \nonumber\\
&\mathbf{w}^l\leq \mathbf{w}\leq \mathbf{w}^U,\nonumber\\
%& \mathbf{y} \in \{0,1\}^q \nonumber \\
& \mathbf{w}_k \equiv \mathbf{w}_i \mathbf{w}_j, \; \forall(i,j,k) \;\in \; \mathcal{T}_{\text{bt}},\nonumber\\
& \mathbf{w}_k \equiv \mathbf{w}_i/\mathbf{w}_j, \; \forall(i,j,k) \;\in \; \mathcal{T}_{\text{lft}}.
\label{eq:exactproblem2}
\end{align}

\begin{figure}
	\centering
	\includegraphics[width=0.5\textwidth]{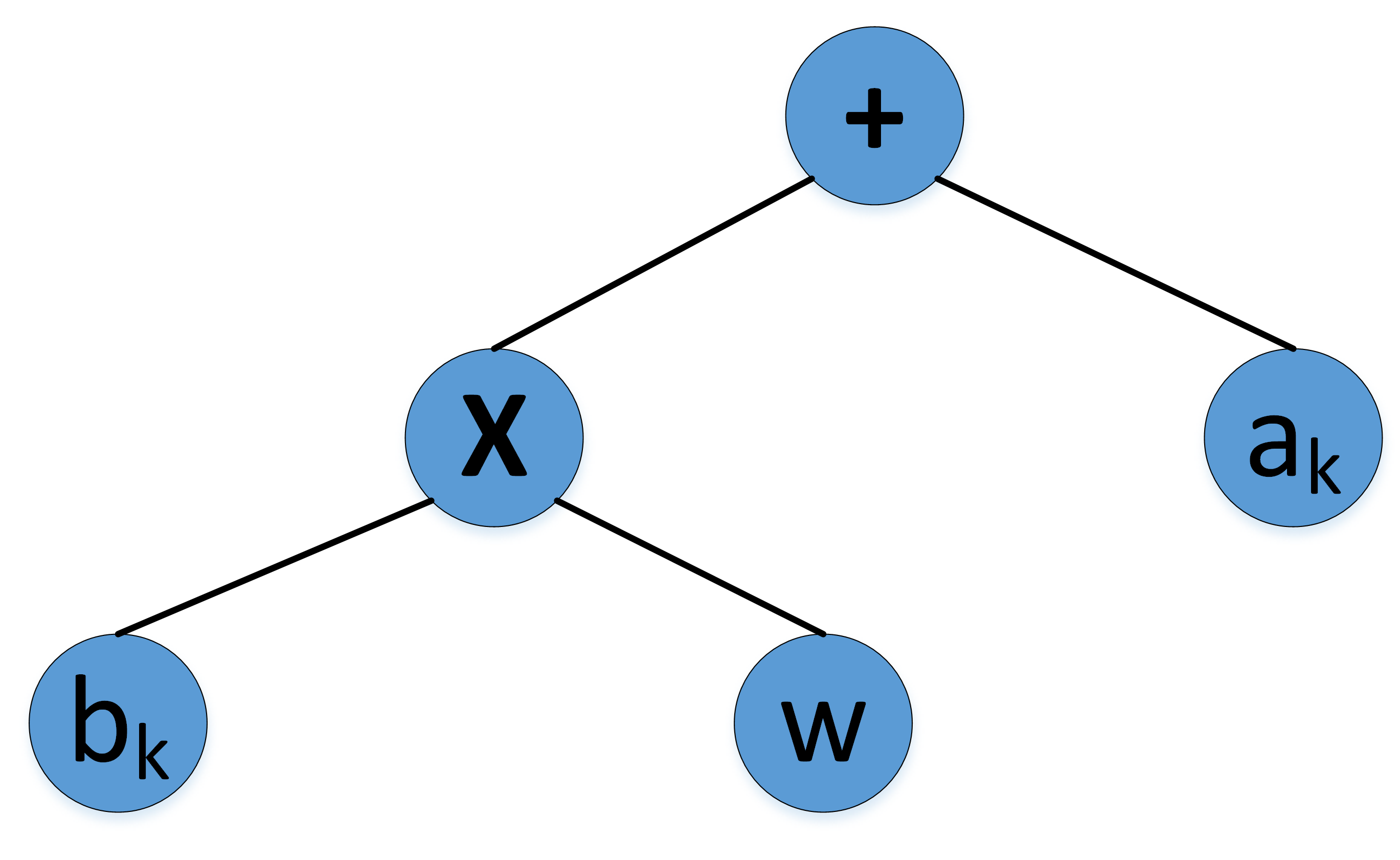}
	\caption{Binary Tree representation for algebraic expression $b_kw+a_k$}
	\protect\label{fig:symbinary}
%	\vspace{-14pt}
\end{figure}

\par All the linear constraints are added into the constraint $Aw=b$ in \eqref{eq:exactproblem2} while the variables introduced are added into the vector $\bf w$ and depending on its type (either bilinear or linear fractional) its definition is added into either $\mathcal{T}_{\text{bt}}$ or $\mathcal{T}_{\text{lft}}$. After such reformulation, we obtain \eqref{eq:exactproblem2}. 
%\red{The sets $\mathcal{T}_{bt}$ and $\mathcal{T}_{lft}$  contain the bilinear, linear fractional and auxiliary terms that 
%where $\mathcal{T}_{bt}$ and $\mathcal{T}_{lft}$ are the resultant sets of bilinear and linear fractional terms, respectively.  
The new variable vector $\mathbf{w}$ consists of continuous and discrete variables in the original MINLP, as well as other auxiliary variables introduced as a result of reformulation. The objective function $w_f$ is a single auxiliary variable.  This reformulation ensures that the new objective function and first constraint in~(\ref{eq:exactproblem2}) are linear, and all non-convexities and non-linearities in the original MINLP are absorbed by the sets $\mathcal{T}_{bt}$ and $\mathcal{T}_{lft}$.

\subsection{Linear Constraint and Variable Creation}
As seen in Table \ref{tab:binaryrules}, certain arithmetic operations during the symbolic reformulation  require creation of new linear constraints and introduction of new variables. This can be easily explained by an example. Let the parent node (any intermediate node that has 2 child nodes) represent a multiplication operation, the left subtree (child) be $abc$ (expression where a,b are constants and c is a variable) and the right subtree be $d$ (variable), then using the rules in the Table \ref{tab:binaryrules}, we need to introduce a new linear constraint (for the left subtree) and then a bilinear variable. So the linear constraint would be $abc-w(i)=0$ where $w(i)$ is the $i^{\text{th}}$ auxiliary variable introduced. The bilinear variable that has to be introduced would be $w(i+1)\equiv w(i)d$ where d is the variable in the original right subtree. The linear constraint will become part of of $Aw=b$ in \eqref{eq:exactproblem2} while $w(i+1)$ will be part of the set of binary terms $\mathcal{T}_{\text{bt}}$
If the intermediate node was  a division operation, then we introduce the linear constraint just like we did for the multiplication operation, however we follow that by adding a linear fractional term $w(i+1)\equiv \frac{w(i)}{d}$. This is a recursive process and is repeated until all the terms in our objective function as well constraints are reformulated. Note that symbolic reformulation does not affect the linear terms in the original problem \eqref{eq:optimization}. Using this process, we reformulate  \eqref{eq:optimization} into \eqref{eq:exactproblem} that can then be used with V-SBB. 
%\newpage
%\vspace{20pt}

\section{}\label{sec:appendixA}
\subsection{Bilinear Terms}
The McCormick linear overestimator and underestimator for bilinear terms with form $w_k \equiv w_i w_j$ are given by \eqref{eq:mcCormickover} and \eqref{eq:mcCormickunder}, respectively.
\begin{align}
w_k \leq w_i^l w_j+w_j^u w_i-w_i^lw_j^u, \nonumber \displaybreak[0]\\
w_k \leq w_i^u w_j+w_j^l w_i-w_i^uw_j^l.
\label{eq:mcCormickover}
\end{align}
\begin{align}
w_k \geq w_i^l w_j+w_j^l w_i-w_i^lw_j^l, \nonumber \displaybreak[0]\\
w_k \geq w_i^u w_j+w_j^u w_i-w_i^uw_j^u.
\label{eq:mcCormickunder}
\end{align}
\subsection{Linear Fractional Terms}
The linear overstimator and underestimator for a linear fractional term withform $w_k \equiv \frac{w_i}{w_j}$ are similar to the overstimator and underestimator of bilinear terms given in \eqref{eq:mcCormickover}  and \eqref{eq:mcCormickunder}.  We first transform the linear fractional term into bilinear term, i.e., $w_i \equiv w_k w_j$ and then we can use \eqref{eq:mcCormickoverlft} and \eqref{eq:mcCormickunderlft} for the the linear overstimator and underestimator, respectively.

\begin{align}
w_i \leq w_k^l w_j+w_j^u w_k-w_k^lw_j^u, \nonumber \displaybreak[0]\\
w_i \leq w_k^u w_j+w_j^l w_k-w_k^uw_j^l.
\label{eq:mcCormickoverlft}
\end{align}
\begin{align}
w_i \geq w_k^l w_j+w_j^l w_k-w_k^lw_j^l,\nonumber \displaybreak[0]\\
w_i \geq w_k^u w_j+w_j^u w_k-w_k^uw_j^u.
\label{eq:mcCormickunderlft}
\end{align}

The advantage of such linear understimator and overestimator is that even if the original problem is a non-convex MINLP,  the relaxed problem will be an MILP which is comparatively easy to solve. 

\section{}\label{sec:appendixB}

For the completeness, we present the following proofs for the convergence of spatial branch-and-bound \cite{liberti2004phdthesis}, which work for our V-SBB. 
\begin{definition}
Let $\Omega \subseteq \mathbb{R}^n$.  A finite family of sets $\mathcal{S}$ is a net for $\Omega$ if it is pairwise disjoint and it covers $\Omega$.	
\end{definition}
\begin{definition}
A net $\mathcal{S}'$ is a refinement of the net $\mathcal{S}$ if there are finitely many pairwise disjoint $s_i' \in \mathcal{S}'$ such that $s = \bigcup_i s_i' \in \mathcal{S}$ and $s \notin S$.
\end{definition}
%In other words, if $\mathcal{S'}$ is a refinement of $\mathcal{S}$, it has been obtained from $\mathcal{S}$ by finitely partitioning some set s in $\mathcal{S}$ and then replacing s by its partitions.  Let $\mathcal{S}_n$ be an infinite sequence of nets for x such that $\forall i \in \mathbb{N}, \mathcal{S}_i$ is a refinement of $\mathcal{S}_{i-1}$.
\begin{definition}
	 Let $\mathcal{M}_n$ be an infinite sequence of subsets of $x$ such that $\mathcal{M}_i \in \mathcal{S}_i$. $\mathcal{M}_n$ is a \emph{filter} for $\mathcal{S}_n$ if $\forall i \in \mathbb{N} \; \mathcal{M}_i\subseteq \mathcal{M}_{i-1}$ where $M_{\infty}= \bigcap_{i\in N}M_i$ be the limit of the filter. 
\end{definition}	
	
\begin{definition}
 Let $x \subseteq \mathbb{R}^n$ and $f(x)$ be the objective function of an MINLP problem then a spatial branch-and-bound algorithm  would be convergent if $\gamma^*=\inf \;f(x)= \lim\limits_{k\to\infty}\gamma_k$	.
\end{definition}	
\begin{definition}
	 A selection rule is \emph{exact} if
	\begin{enumerate}
		\item  The infinimum objective function value of any region that remains qualified during the whole
		solution process is greater than or equal to the globally optimal objective function value, i.e., 
		\begin{equation*}
		\forall M \in \bigcap_{k=1}^{\infty}\mathcal{R}_k(\inf\; f(x\cap M) \geq \gamma^*)
		\end{equation*}
		\item The limit $\mathbb{M}_{\infty}$ of any filter $M_k$ is such that$\inf f(\Omega \cap M)\geq \gamma^*$ where $\Omega$ is the feasible set.
	\end{enumerate}
	\end{definition}	
    
     \begin{theorem}
    	A Spatial branch-and-bound algorithm using an exact selection rule converges.
    \end{theorem} 
    \begin{proof}{Proof by contradiction:}
    	
    	Let there be $x \in \Omega$ with $f(x) < \gamma^*$. Let $x \in M$ with $M \in \mathbb{R}_n$ for some $n \in \mathbb{N}$. Because of the first condition of exactness of selection rule, the filter $M$ cannot remain qualified forever.  Furthermore, unqualified regions may not, by hypothesis, include points with better objective function values than the current incumbent $\gamma_k$. Hence $M$ must necessarily be split at some iteration $n'>n$ so $x$ belongs to every $\mathcal{M}_n$ in some filter $\{\mathcal{M}_n\}$, thus $x \in \Omega \cap \mathbb{M}_{\infty}$. By condition $2$ of exactness of selection rule, $f(x)\geq f(\Omega \cap M_\infty)\geq \gamma^*$. The result follows.
    \end{proof}

%\section{}\label{sec:appendixg}
%\input{13-appG}

%\section{}\label{sec:appendixf}
%\input{12-appF}

\end{document}